\documentclass[preprint]{imsart}
\pdfoutput=1

\usepackage{amsfonts}
\usepackage{amsmath}
\usepackage{amsthm}
\usepackage{graphicx}

\def \be{\begin{equation}}
\def \ee{\end{equation}}
\newlength\figheight
\newlength\figwidth

\usepackage{epsfig} 
\usepackage{amssymb} 
\usepackage{amsmath} 
\usepackage{amsthm}
\usepackage{comment}

\usepackage{todonotes}
\usepackage{graphics,graphicx}
\usepackage{eepic}
\usepackage{color}
\usepackage{epsfig}
\usepackage{amscd}
\usepackage{subfig}
\usepackage{todonotes}

\definecolor{darkred}{rgb}{.7,0,0}

\definecolor{darkgreen}{rgb}{0,0.7,0}

\definecolor{darkblue}{rgb}{0,0,0.7}

\newcommand{\us}{\mathsf{U}}
\newcommand{\rs}{\mathsf{r}}
\newcommand{\bp}{\begin{proof}}
\newcommand{\ep}{\end{proof}}
\newcommand{\bbR}{\mathbb{R}}
\newcommand{\cV}{\mathcal{V}}
\newcommand{\phs}{\phi^{\mathsf{s}}}
\newcommand{\phr}{\phi^{\mathsf{r}}}
\newcommand{\vs}{v^{\mathsf{s}}}
\newcommand{\vr}{v^{\mathsf{r}}}
\newcommand{\vars}{\varphi^{\mathsf{s}}}
\newcommand{\varr}{\varphi^{\mathsf{r}}}

\newcommand{\Uhs}{U^{h, \sigma}}
\newcommand{\Uhz}{U^{h, 0}}
\newcommand{\Phs}{\mathbb{P}^{h, \sigma}}
\newcommand{\Phz}{\mathbb{P}^{h, 0}}

\newcommand{\psh}{\Psi_{h}}
\newcommand{\pss}{\Psi_{\tau}}
\newcommand{\pst}{\Psi_{s-t_k}}

\newcommand{\gh}{g^h}
\newcommand{\hC}{C^h}

\newtheorem{theo}{Theorem}[section]
\newtheorem{theorem}[theo]{Theorem}

\newtheorem{remark}[theo]{Remark}

\newtheorem{example}[theo]{Example}

\def\cV{\mathcal{V}}

\def\qed{\hfill $\square$ \goodbreak \bigskip}

\newtheorem{assumption}{Assumption}
\newcommand{\bes}{\begin{equation*}}
\newcommand{\ees}{\end{equation*}}
\newcommand{\bea}{\begin{eqnarray}}
\newcommand{\eea}{\end{eqnarray}}

\newcommand{\R}{\mathbb{R}}

\newcommand{\hL}{{\mathcal L}^h}
\newcommand{\oL}{{\mathcal L}}
\newcommand{\aL}{\widehat{{\mathcal L}}^h}
\newcommand{\tL}{\widetilde{{\mathcal L}}^h}

\newcommand{\cF}{\mathcal{F}}
\newcommand{\E}{\mathbb{E}^h}

\def \IE{\mathbb E}

\DeclareMathOperator*{\argmax}{arg\,max}

\newcommand{\LRARR}[4]{{\mbox{\raise 0.4 mm \hbox{$#1$}}} \;
  \mathop{\stackrel{\displaystyle\longrightarrow}\longleftarrow}^{#3}_{#4}
  \; {\mbox{\raise 0.4 mm\hbox{$#2$}}}}

\RequirePackage[OT1]{fontenc}
\RequirePackage{amsthm,amsmath}
\RequirePackage[numbers]{natbib}
\RequirePackage[colorlinks,citecolor=blue,urlcolor=blue]{hyperref}

\arxiv{arXiv:0000.0000}

\startlocaldefs
\numberwithin{equation}{section}
\theoremstyle{plain}

\endlocaldefs

\begin{document}

\begin{frontmatter}
\title{Probability Measures for Numerical Solutions of Differential Equations 
}
\runtitle{Probabilistic Solutions of Differential Equations}

\begin{aug}
\author{\fnms{Patrick} \snm{Conrad}\thanksref{t1}\ead[label=e1]{p.conrad@warwick.ac.uk}},
\author{\fnms{Mark} \snm{Girolami}\thanksref{t2}\ead[label=e2]{m.girolami@warwick.ac.uk}},
\author{\fnms{Simo} \snm{S\"{a}rkk\"{a}}\thanksref{t4}\ead[label=e5]{simo.sarkka@aalto.fi}},
\author{\fnms{Andrew} \snm{Stuart}\thanksref{t3}\ead[label=e3]{a.m.stuart@warwick.ac.uk}}
\and
\author{\fnms{Konstantinos} \snm{Zygalakis}\ead[label=e4]{k.zygalakis@soton.ac.uk}}



\thankstext{t1}{Supported by EPSRC grant CRiSM EP/D002060/1}
\thankstext{t2}{Supported by EPSRC Established Career Research Fellowship EP/J016934/2 }
\thankstext{t3}{Supported by EPSRC Programme Grant EQUIP EP/K034154/1}
\thankstext{t4}{Supported by Academy of Finland Research Fellowship 266940}
\runauthor{P. Conrad et al.}

\affiliation{University of Warwick\thanksmark{m1}, University of Southampton\thanksmark{m2}, Aalto University\thanksmark{m3}}

\address{Department of Statistics\\University of Warwick\\
Coventry, CV4 7AL\\United Kingdom\\
\printead{e1}\\
\phantom{E-mail:\ }\printead*{e2}}

\address{Mathematics Institute \\University of Warwick\\
Coventry, CV4 7AL\\United Kingdom\\
\printead{e3}\\
}

\address{Department of Mathematics\\
University of Southampton\\
Southampton, SO17 1BJ\\United Kingdom\\
\printead{e4}\\
}

\address{ Department of Electrical Engineering and Automation\\
Aalto University\\
02150 Espoo\\
Finland\\
\printead{e5}\\
}
\end{aug}

\begin{abstract}
In this paper, we present a formal quantification of epistemic uncertainty induced by numerical solutions of ordinary and partial differential equation models.
Numerical solutions of differential equations contain inherent uncertainties due to the finite dimensional approximation of an unknown and implicitly defined function. When statistically analysing models based on differential equations describing physical,  or other naturally occurring, phenomena, it is therefore important to explicitly account for the uncertainty introduced by the numerical method. This enables objective determination of its importance relative to other uncertainties, such as those caused by data contaminated with noise or model error induced by missing physical or inadequate descriptors. To this end we show that a wide variety of existing solvers can be randomised, inducing a probability measure over the solutions of such differential equations. These measures exhibit contraction to a Dirac measure around the true unknown solution, where the rates of convergence are consistent with the underlying deterministic numerical method. 
Ordinary differential equations and elliptic partial differential equations are used to illustrate the approach to quantifying uncertainty in both the statistical analysis of the forward and inverse problems.
\end{abstract}



\end{frontmatter}


\section{Introduction and Motivation}

The numerical analysis literature has developed a large range of efficient algorithms for solving ordinary and partial differential equations, which are typically designed to solve a single problem as efficiently as possible  \cite{hairer1993solving,eriksson1996computational}. When classical numerical methods are placed within statistical analysis, however, we argue that significant difficulties can arise.
Statistical models are commonly constructed by naively replacing the analytical solution with a numerical method; these models are then typically analysed by computational methods that sample the solver with many slightly varying configurations. This is problematic, not only because the numerical solutions are merely approximations,  but also because the assumptions made by the numerical method induce correlations across multiple numerical solutions.
 While the distributions of interest commonly do converge asymptotically as the solver mesh becomes dense (e.g., in statistical inverse problems \cite{DS15}), we argue that at a finite resolution the analysis may be vastly overconfident as a result of these unmodelled errors.

The purpose of this paper is to address these issues by the construction and rigorous analysis of novel probabilistic integration methods
for both ordinary and partial differential equations. The approach in both cases is similar: we identify the key discretisation assumptions and introduce a local random field, in particular a Gaussian field, to reflect our uncertainty in those assumptions. The probabilistic solver may then be sampled repeatedly to interrogate the uncertainty in the solution. For a wide variety of commonly used numerical methods, our construction is straightforward to apply and preserves the formal order of convergence of the original method.

Furthermore, we demonstrate the value of these probabilistic solvers in statistical inference settings. 
Analytic and numerical examples show that using a classical non-probabilistic solver with inadequate discretisation when performing inference can lead to inappropriate and misleading posterior concentration in a Bayesian setting. In contrast, the probabilistic solver reveals the structure of uncertainty in the solution, naturally limiting posterior concentration as appropriate.

Our strategy fits within the emerging field known as Probabilistic Numerics \cite{Hennig2015}, a perspective on computational methods pioneered by 
Diaconis \cite{diaconis1988bayesian}, and subsequently 
Skilling \cite{skilling1992bayesian}. 
Under this framework, solving differential equations is recast as the statistical inference problem of finding the latent, unobserved function that solves the differential equation, based on observable constraints.
The output from such an inference problem is a probability measure over functions that satisfy the constraints imposed by the specific differential equation. This measure formally quantifies the uncertainty in candidate solution(s) of the differential equation and allows us to coherently propagate this uncertainty when carrying out statistical analyses, for example, in
uncertainty quantification \cite{S15} or Bayesian inverse problems \cite{DS15}.

\subsection{Review of existing work}
Earlier work in the numerical analysis literature including randomisation in the approximate integration of ordinary differential equations (ODEs) includes \cite{coul99,sten95}. However, the
connections between probability and numerical analysis were elucidated in the foundational paper \cite{diaconis1988bayesian}, which explains this connection in
the context of Bayesian statistics and sets out an agenda for Probabilistic Numerics. It was in \cite{skilling1992bayesian} that a concrete
demonstration of how numerical integration of ordinary differential
equations could be interpreted as a Bayesian inference method,
and this idea was converted into an operational methodology,
with consistency proofs given in \cite{chkrebtii2013bayesian}.

The methodology detailed in \cite{skilling1992bayesian, chkrebtii2013bayesian} (explored in parallel in \cite{Hennig2014})
builds a probabilistic structure which contains two sources of randomness: an overarching Gaussian
process structure and an internal randomisation of the mean process, resulting
in different draws of the mean for each realisation. However, careful
analysis reveals that, because of the conjugacy inherent in linear Gaussian models, 
the posterior variance of the outer Gaussian process is independent of the observed constraints, and thus does not reveal any information about the uncertainty in the particular problem being solved. The randomisation of the mean process is more useful, and is related to the strategy we present here. Those initial efforts were limited to first order Euler schemes \cite{hairer1993solving}, and were later extended to Runge--Kutta methods \cite{schober2014probabilistic}, although the incomplete characterisation of uncertainty remains. The existing literature has focussed only on ordinary differential equations and a formal probabilistic construction for partial differential equations (PDEs) that is generally applicable remains an open issue. All of these open issues are resolved by this work. Unlike \cite{skilling1992bayesian, chkrebtii2013bayesian, schober2014probabilistic}, which are exclusively Bayesian by construction, our probabilistic solvers return a probability measure which can then be employed in either frequentist or Bayesian inference frameworks.

Our motivation for enhancing inference problems with models of discretisation error is similar to the more general concept of model error, as developed by \cite{kennedy2001bayesian}. However, our focus on errors arising from the discretisation of differential equations leads to more specialised methods. Existing strategies for discretisation error include empirically fitted Gaussian models for PDE errors \cite{kaipio2007statistical} and randomly perturbed ODEs \cite{arnold2013linear}; the latter partially coincides with our construction, but our motivation and analysis is distinct. Recent work \cite{capistran2013bayesian} uses Bayes factors to analyse the impact of discretisation error on posterior approximation quality. Probabilistic models have also been used to study error propagation
due to rounding error; see \cite{hairer2008achieving}.

The remainder of the paper has the following structure: Section \ref{sec:ode} introduces and formally analyses the proposed probabilistic solvers for ODEs. 
Section \ref{sec:numerics} explores the characteristics of random solvers employed in the statistical analysis of both forward and inverse problems. 
Then, we turn to elliptic PDEs in Section \ref{sec:pde}, where several key steps of the construction of probabilistic solvers and their analysis have intuitive analogues in the ODE context. Finally, an illustrative example of PDE inference problem is presented in Section \ref{sec:pdenumerics}.\footnote{Supplementary materials and code are available online: \url{http://www2.warwick.ac.uk/fac/sci/statistics/staff/academic-research/girolami/probints/}}

\section{Probability Measures via Probabilistic Time Integrators}
\label{sec:ode}


Consider the following ordinary differential equation (ODE):
\begin{equation} \label{eq:ODE}
\frac{du}{dt}=f(u), \quad u(0)=u_{0},
\end{equation}
where $u(\cdot)$ is a continuous function taking values in $\R^n$.\footnote{To simplify our discussion we assume that the ODE is autonomous, that is, $f(u)$ is independent of time. Analogous theory can be developed for time dependent forcing.}
We let $\Phi_t$ denote the flow-map for Equation \eqref{eq:ODE}, so that
$u(t)=\Phi_t\bigl(u(0)\bigr)$. The conditions ensuring that this solution exists will be formalised in Assumption \ref{a:2}, below.

Deterministic numerical methods for the integration of this equation
on time-interval $[0,T]$ will produce an approximation to the equation on a mesh of
points $\{t_k=kh\}_{k=0}^{K}$, with $Kh=T$, (for simplicity we assume
a fixed mesh). Let $u_k=u(t_k)$ denote the exact solution of
\eqref{eq:ODE} on the mesh and $U_k\approx u_k$ denote the approximation computed by the numerical method, based on evaluating  $f$ (and possibly higher derivatives)
at a finite set of points that are generated during numerical
integration.  
Typically these methods output a single discrete solution $\{U_k\}_{k=0}^K$, possibly augmented with some type of error indicator, but do not statistically quantify the uncertainty remaining in the path.

%
Let $X_{a,b}$ denote the Banach space $C([a,b];\R^n)$. 
The exact solution of \eqref{eq:ODE} on the time-interval $[0,T]$
may be viewed as a Dirac
measure $\delta_u$ on $X_{0,T}$ at the element $u$ that 
solves the ODE. We will construct a probability measure
$\mu^h$ on $X_{0,T}$, that is straightforward to sample from both on and
off the mesh, for
which $h$ quantifies the size of the discretisation step employed, and whose 
distribution reflects the uncertainty resulting from the solution of the ODE. 
Convergence of the numerical method is then related to the contraction
of $\mu^h$ to $\delta_u$.

We briefly summarise the construction of the numerical method.
Let $\psh:\R^n \to \R^n$ denote a classical deterministic
one-step numerical integrator
over time-step $h$, a class including all Runge--Kutta methods
and Taylor methods for ODE numerical integration \cite{hairer1993solving}.
Our numerical methods will have the property that, on the mesh,
they take the form
\begin{equation}
U_{k+1}=\psh(U_k)+\xi_k(h),
\label{eq:NM}
\end{equation}
where $\xi_k(h)$ are suitably scaled, i.i.d. Gaussian random variables. That is, the random solution iteratively takes the standard step, $\psh$, followed by perturbation with a random draw, $\xi_k(h)$, modelling uncertainty that accumulates between mesh points. The discrete path $\{U_k\}_{k=0}^K$ is straightforward to sample and in general is not a Gaussian process. Furthermore, the discrete trajectory can be extended into a continuous time approximation of the 
ODE, which we define as a draw from the measure $\mu^h$. 

The remainder of this section develops these solvers in detail and proves strong convergence of the random solutions to the exact solution, implying that $\mu^h \to \delta_u$ in an appropriate sense. Finally, we establish a close relationship between our random solver and a stochastic differential equation (SDE) with small mesh-dependent noise.

\subsection{Probabilistic time-integrators: general formulation}
\label{sec:ode:solver}
The integral form of Equation \eqref{eq:ODE} is 
\begin{equation}
\label{eq:IE}
u(t)=u_0+\int_{0}^{t}f\bigl(u(s)\bigr)ds. 
\end{equation}
The solutions on the mesh satisfy 
\begin{equation}
\label{eq:1}
u_{k+1}=u_k+\int_{t_k}^{t_{k+1}}f\bigl(u(s)\bigr)ds,
\end{equation}
and may be interpolated between mesh points by means of the expression 
\begin{equation}
\label{eq:2}
u(t)=u_k+\int_{t_k}^{t}f\bigl(u(s)\bigr)ds, \quad t \in [t_k,t_{k+1}).  
\end{equation}
We may then write
\begin{equation}
u(t)=u_{k}+\int_{t_k}^{t}g(s)ds, \quad t \in [t_k,t_{k+1}),
\label{eq:33}
\end{equation}
where $g(s)=f\bigl(u(s)\bigr)$ is an unknown function of time. In
the algorithmic setting we have approximate knowledge about $g(s)$
through an underlying numerical method. A variety of traditional 
numerical algorithms may be derived based on approximation of $g(s)$ by 
various simple deterministic functions $\gh(s)$. 
Perhaps the simplest 
such numerical method arises from invoking the Euler approximation that
\begin{equation}
\gh(s)=f(U_{k}), \quad s \in [t_k,t_{k+1}).
\label{eq:d}
\end{equation}
In particular, if we take $t=t_{k+1}$ and apply this method inductively
the corresponding numerical scheme arising from making such an approximation 
to $g(s)$ in \eqref{eq:33} is $U_{k+1}=U_{k}+hf(U_{k}).$ Now consider the more general one-step numerical method  $U_{k+1}=\psh(U_k).$
This may be derived by approximating $g(s)$ in \eqref{eq:33} by
\begin{equation}
\gh(s)=\frac{d}{d\tau}\Bigl(\pss(U_k)\Bigr)_{\tau=s-t_k}, \quad s \in [t_k,t_{k+1}).
\label{eq:dg}
\end{equation}
We note that all consistent (in the sense of numerical analysis)
one-step methods will satisfy
$$\frac{d}{d\tau}\Bigl(\pss(u)\Bigr)_{\tau=0}=f(u).$$

The approach based on the approximation \eqref{eq:dg}
leads to a deterministic numerical method which is
defined as a continuous function of time. Specifically we
have $U(s)=\Psi_{s-t_k}(U_k), \quad s \in [t_k,t_{k+1}).$ Consider again the Euler approximation, 
for which $\pss(U)=U+\tau f(U)$, and whose continuous time
interpolant is then given by
$U(s)=U_k+(s-t_k)f(U_k), \;\; s \in [t_k,t_{k+1}).$
Note that this produces a continuous function, namely an element
of $X_{0,T}$, when extended to $s \in [0,T].$
The preceding development of a numerical integrator does not
acknowledge the uncertainty that arises from lack of knowledge
about $g(s)$ in the interval $s \in [t_k,t_{k+1}).$ 
We propose to approximate $g$ 
stochastically in order to represent this uncertainty, taking
\[
\gh(s)= \frac{d}{d\tau}\Bigl(\pss(U_k)\Bigr)_{\tau=s-t_k}+\chi_k(s-t_k), \quad s \in [t_k,t_{k+1})
\]
where the $\{\chi_k\}$ form an i.i.d. sequence of Gaussian random
functions defined on $[0,h]$ with $\chi_k \sim N(0,\hC)$\footnote{We use $\chi_k \sim N(0,\hC)$ to denote a zero-mean Gaussian process defined on $[0,h]$ with a covariance kernel $\mathrm{cov}(\chi_k(t),\chi_k(s)) \triangleq C^h(t,s)$.}.

We will choose
$\hC$ to  shrink to zero with $h$ at a prescribed rate, and
also to ensure that $\chi_k \in X_{0,h}$ almost surely. 
The functions $\{\chi_k\}$ represent our uncertainty about the function $g$.  
The corresponding numerical scheme arising from such an approximation 
is given by 
\begin{equation} \label{eq:probg}
U_{k+1}=\psh(U_{k})+\xi_k(h)
\end{equation}
where the i.i.d. sequence of functions $\{\xi_k\}$ lies in $X_{0,h}$ and is
given by
\begin{equation}
\xi_k(t)=\int_{0}^{t}\chi_{k}(\tau)d\tau.
\end{equation}
Note that the numerical solution is now naturally defined between grid points,
via the expression
\begin{equation} \label{eq:probg2}
U(s)=\pst(U_{k})+\xi_k(s-t_k), \quad s \in [t_k,t_{k+1}).
\end{equation}

In the case of the Euler method, for example, we have
\begin{equation} \label{eq:prob}
U_{k+1}=U_{k}+hf(U_{k})+\xi_k(h)
\end{equation}
and, between grid points,
\begin{equation} \label{eq:prob2}
U(s)=U_{k}+(s-t_k)f(U_{k})+\xi_k(s-t_k), \quad s \in [t_k,t_{k+1}).
\end{equation}
This method is illustrated in Figure \ref{fig:EulerCartoon}.

\begin{figure}[htb]
\centering
\subfloat[Deterministic Euler.]{
\includegraphics[scale=.6]{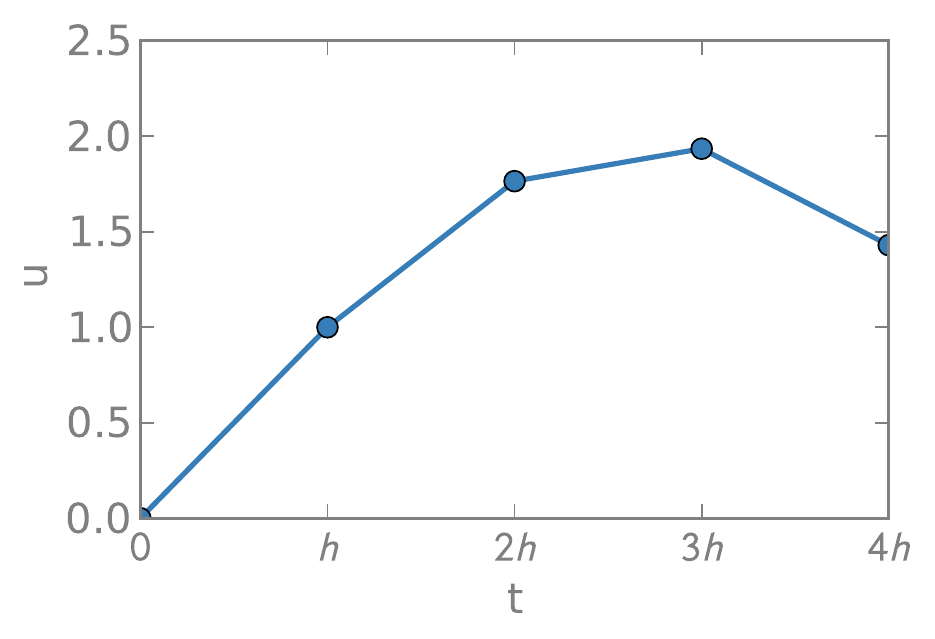}
}
\subfloat[Randomized Euler.]{
\includegraphics[scale=.6]{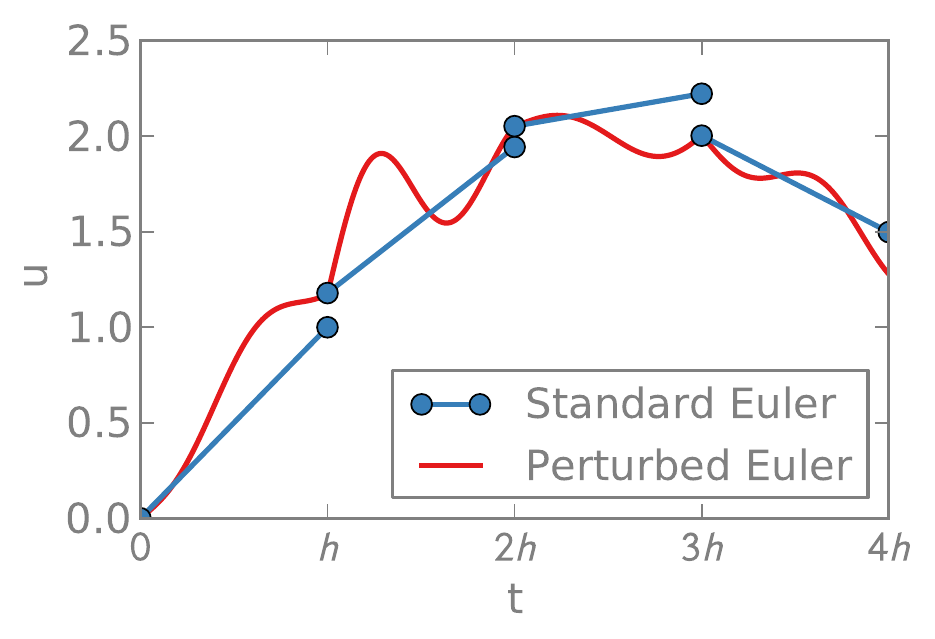}
}
\caption{An illustration of deterministic Euler steps and randomised variations. The random integrator in (b) outputs the path in red; we overlay the standard Euler step constructed at each step, before it is perturbed (blue).}
\label{fig:EulerCartoon}
\end{figure}

While we argue that the choice of modelling local uncertainty in the flow-map as a Gaussian process is natural and analytically favourable, it is not unique. 
It is possible to construct examples where the Gaussian assumption is invalid; for example, when a highly inadequate time step is used, a systemic bias may be introduced. However, in regimes where the underlying deterministic 
method performs well, the centred Gaussian assumption is a reasonable prior.

Finally, we comment on sampling off-grid. Imagine that we wish to
sample the continuous time solution at point $s \in (t_k,t_{k+1}).$
This might be performed in two different ways, given the solution generated on-grid, up to time $t_k$, producing
sample $U_k$. In the first method, we draw the sample
$U(s)=\Psi_{s-t_k}(U_k)+\xi_k(s-t_k).$
Having done so we must be careful to draw
$U_{k+1}=\Psi_{h}(U_k)+\xi_k(h) \mid \xi_k(s-t_k).$
This is not complicated, because all random variables are Gaussian given available information,
but needs to be respected via conditioning on the end points of the grid. An alternative is to draw the sample
$U_{k+1}=\Psi_{h}(U_k)+\xi_k(h)$
first, followed by the sample
$U(s)=\Psi_{s-t_k}(U_k)+\xi_k(s-t_k)\mid \xi_k(h),$
which the conditional Gaussian structure facilitates.

\subsection{Strong convergence result}
\label{sec:ode:convergence}

To prove the strong convergence of our probabilistic numerical solver, we first need two assumptions quantifying properties of the random noise and of the underlying deterministic integrator, respectively.
In what follows we use $\langle \cdot, \cdot \rangle$ and $|\cdot|$
to denote the Euclidean inner-product and norm on $\R^n$. We denote
the Frobenius norm on $\R^{n \times n}$ by $|\cdot|_{\rm F}$, 
and $\E$ denotes expectation with respect to the i.i.d.
sequence $\{\chi_k\}$. 
\begin{assumption}
\label{a:1}
Let $\xi_k(t):=\int_0^t \chi_k(s)ds$ with $\chi_k \sim N(0,\hC).$
Then there exists $K>0, p \ge 1$ such that, for all $t \in [0,h]$,
$\E |\xi_k(t) \xi_k(t)^T|_{\rm F}^2 \le Kt^{2p+1};$ in particular
$\E |\xi_k(t)|^2 \le Kt^{2p+1}.$ Furthermore we assume the existence
of matrix $Q$, independent of $h$, such that $\E [\xi_k(h) \xi_k(h)^T]=Qh^{2p+1}.$
\end{assumption}

Here, and in the sequel, $K$ is a constant independent of $h$, but
possibly changing from line to line.
The covariance kernel $\hC$ is not uniquely defined, but the uncertainty at the beginning of the step must be zero and increase with step length, scaling exactly as $h^{2p+1}$ at the end of the interval. Implementation requires selection of the constant matrix $Q$; in our numerical examples
we assume a scaled identity matrix, $Q=\sigma I$, and we discuss one possible strategy for choosing $\sigma$ in Section \ref{sec:numerics:calibration}. 

\begin{assumption}
\label{a:2}
The function $f$ and a sufficient number of its derivatives
are bounded uniformly in $\R^n$ in order to ensure that
$f$ is globally Lipschitz
and that the numerical flow-map $\Psi_h$ has uniform local truncation
error of order $q+1$:
$$\sup_{u \in \R^n} |\Psi_t(u)-\Phi_t(u)| \le Kt^{q+1}.$$
\end{assumption}

\begin{remark}
\label{rem:2.1}
We assume globally Lipschitz $f$, and bounded derivatives, in order
to highlight the key probabilistic ideas, whilst simplifying the numerical 
analysis. Future work will address the non-trivial issue
of extending of analyses to weaken these assumptions. In this
paper, we provide numerical results indicating that a weakening of
the assumptions is indeed possible. 
\end{remark}

\begin{theorem}
\label{t:1}
Under Assumptions \ref{a:1},\ref{a:2} it follows that there is $K>0$
such that 
$$\sup_{0 \le kh \le T}\E|u_k-U_k|^2 \le Kh^{2\min\{p,q\}}.$$
Furthermore
$$\sup_{0 \le t \le T}\E|u(t)-U(t)| \le Kh^{\min\{p,q\}}.$$
\end{theorem}

This theorem implies that every probabilistic solution is a good approximation of the exact solution in both a discrete and continuous sense. Choosing $p \ge q$ is natural if we want to preserve the
strong order of accuracy of the underlying deterministic
integrator; we proceed with the choice $p=q$, introducing
the maximum amount of noise consistent
with this constraint.

\begin{proof}
We first derive the convergence result on the grid, and then in continuous time.
From \eqref{eq:probg} we have 
\begin{equation} \label{eq:probga}
U_{k+1}=\psh(U_{k})+\xi_k(h)
\end{equation}
whilst we know that
\begin{equation} \label{eq:probgb}
u_{k+1}=\Phi_h(u_{k}).
\end{equation}
Define the truncation error
$\epsilon_k=\psh(U_{k})-\Phi_h(U_k)$
and note that
\begin{equation} \label{eq:probgc}
U_{k+1}=\Phi_h(U_k)+\epsilon_k+\xi_k(h).
\end{equation}
Subtracting equation \eqref{eq:probgc} from \eqref{eq:probgb}
and defining $e_k=u_k-U_k$, we get
$$e_{k+1} = \Phi_h(u_{k})-\Phi_h(u_{k}-e_{k})-\epsilon_k-\xi_k(h).$$
Taking the Euclidean norm and taking expectations gives, using
Assumption \ref{a:1} and the independence of the $\xi_k$, 
$$\E|e_{k+1}|^2 =\E\Bigl|\Phi_h(u_{k})-\Phi_h(u_{k}-e_{k})-\epsilon_k\Bigr|^2+{\mathcal O}(h^{2p+1})$$
where the constant in the ${\mathcal O}(h^{2p+1})$ term is uniform in 
$k: 0 \le kh \le T.$ Assumption \ref{a:2} implies that 
$\epsilon_k={\mathcal O}(h^{q+1}),$ again uniformly in $k: 0 \le kh \le T.$
Noting that $\Phi_h$ is globally Lipschitz with constant bounded
by $1+Lh$ under Assumption \ref{a:2}, we then obtain
\begin{eqnarray*}
\E|e_{k+1}|^2 &\le& (1+Lh)^2 \E|e_k|^2+ \E\Bigl|\bigl\langle h^{\frac12}\bigl(\Phi_h(u_{k})-\Phi_h(u_{k}-e_{k})\bigr),h^{-\frac12}\epsilon_k
\bigr\rangle \Bigr|\\&+&{\mathcal O}(h^{2q+2})+{\mathcal O}(h^{2p+1}).
\end{eqnarray*}
Using Cauchy--Schwarz on the inner-product, and the fact that $\Phi_h$
is Lipschitz with constant bounded independently of $h$, we get
$$\E|e_{k+1}|^2 \le \bigl(1+{\mathcal O}(h)\bigr)\E|e_k|^2+
{\mathcal O}(h^{2q+1})+{\mathcal O}(h^{2p+1}).$$
Application of the Gronwall inequality gives the desired result.

Now we turn to continuous time. We note that, for $s \in [t_k,t_{k+1})$, 
\begin{align*}
U(s)&=\Psi_{s-t_k}(U_k)+\xi_k(s-t_k),\\
u(s)&=\Phi_{s-t_k}(u_k).
\end{align*}
Let $\cF_t$ denote the $\sigma$-algebra of events generated by the
$\{\xi_k\}$ up to time $t$.
Subtracting we obtain, using Assumptions \ref{a:1} and \ref{a:2} 
and the fact that $\Phi_{s-t_k}$ has Lipschitz constant of the form 
$1+{\mathcal O}(h)$,
\begin{align*}
\E\bigl(|U(s)-u(s)|\big| \cF_{t_k}\bigr) &\le  |\Phi_{s-t_k}(U_k)-\Phi_{s-t_k}(u_k)|+ |\Psi_{s-t_k}(U_k)-\Phi_{s-t_k}(U_k)|\\&\hspace{1.78in}+\E \bigl(|\xi_k(s-t_k)|\big|\cF_{t_k}\bigr)\\
&\le (1+Lh)|e_k|+{\mathcal O}(h^{q+1})+\E |\xi_k(s-t_k)|\\
&\le (1+Lh)|e_k|+{\mathcal O}(h^{q+1})+\bigl(\E|\xi_k(s-t_k)|^2\bigr)^{\frac12}\\
&\le (1+Lh)|e_k|+{\mathcal O}(h^{q+1})
+{\mathcal O}(h^{p+\frac12}).
\end{align*}
Now taking expectations we obtain
$$\E|U(s)-u(s)| \le (1+Lh)\bigl(\E|e_k|^2\bigr)^{\frac12}+{\mathcal O}(h^{q+1})
+{\mathcal O}(h^{p+\frac12}).$$
Using the on-grid error bound gives the desired result, after noting
that the constants appearing are uniform in $0 \le kh \le T.$
\end{proof}

\subsection{Examples of probabilistic time-integrators}
\label{sec:ode:examples}

The canonical illustration of a probabilistic time-integrator
is the probabilistic Euler method already described. 
In this section we describe two methods which showcase the
generality of the approach. 


The first is the classical {\em Runge--Kutta method} which
defines a one-step numerical integrator as follows: 
$$\psh(u)=u+\frac{h}{6}\bigl(k_1(u)+2k_2(u,h)+2k_3(u,h)+k_4(u,h)\bigr)$$
where 
\begin{subequations}
\begin{align*}
k_1(u)&=f(u),\quad
k_2(u,h)=f\bigl(u+\frac12h k_1(u)\bigr)\\
k_3(u,h)&=f\bigl(u+\frac12h k_2(u)\bigr),\quad
k_4(u,h)=f\bigl(u+h k_3(u)\bigr).
\end{align*}
\end{subequations}
The method has local truncation error in the form of Assumption \ref{a:2}
with $q=4.$ It may be used as the basis of
a probabilistic numerical method \eqref{eq:probg2}, and hence
\eqref{eq:probg} at the grid-points. 
Thus, provided that we choose to perturb this integrator with a random process
$\chi_k$ satisfying Assumption \ref{a:1} with $p \ge 4$, Theorem
\ref{t:1} shows that the error between the probabilistic
integrator based on the classical Runge--Kutta method is, in the mean square 
sense, of the same order of accuracy as the deterministic classical Runge--Kutta
integrator.

The second is an {\em integrated Ornstein--Uhlenbeck process},
derived as follows. Define,
on the interval $s \in [t_k,t_{k+1})$, the pair of equations
\begin{subequations}
\label{eq:OU2}
\begin{align}
dU&=Vdt, \quad U(t_k)=U_k,\\
dV&=-\Lambda Vdt+\sqrt{2\Sigma} \, dW, \quad V(t_k)=f(U_k).
\end{align}
\end{subequations}
Here $W$ is a standard Brownian motion and
$\Lambda$ and $\Sigma$ are invertible matrices, possibly depending on $h$.
The approximating function $\gh(s)$ is thus defined by $V(s)$,
an Ornstein--Uhlenbeck process. 

Integrating (\ref{eq:OU2}b) we obtain
\begin{equation}
\label{eq:6}
V(s)=\exp\bigl(-\Lambda(s-t_k)\bigr)f(U_{k})+\chi_k(s-t_k), 
\end{equation}
where $\quad s \in [t_k,t_{k+1})$ and the $\{\chi_k\}$ form an i.i.d. sequence of Gaussian random
functions defined on $[0,h]$ with 
$$\chi_k(s)=\sqrt{2\Sigma}\int_0^{s}\exp\bigl(\Lambda(\tau-s)\bigr) \, dW(\tau).$$
Note that the $h$-dependence of $\hC$ comes through
the time-interval on which $\chi_k$ is defined, and through $\Lambda$
and $\Sigma$.

Integrating (\ref{eq:OU2}a), using \eqref{eq:6}, we obtain
\begin{equation}
\label{eq:probgex}
U(s)=U_k+\Lambda^{-1}\Bigl(I-\exp\bigl(-\Lambda (s-t_k)\bigr)\Bigr)f(U_{k})+\xi_k(s-t_k),
\end{equation}
where $\quad s \in [t_k,t_{k+1}]$, and, for $t \in [0,h]$,
\begin{equation}
\label{eq:chi2}
\xi_k(t)=\int_{0}^{t}\chi_{k}(\tau)d\tau.
\end{equation}
The numerical method \eqref{eq:probgex} 
may be written in the form \eqref{eq:probg2}, and hence
\eqref{eq:probg} at the grid-points, with the definition
$$\psh(u)=u+\Lambda^{-1}\Bigl(I-\exp\bigl(-\Lambda h\bigr)\Bigr)f(u).$$
This integrator is first order accurate and satisfies
Assumptions \ref{a:2} with $p=1$. Choosing to scale $\Sigma$ with
$h$ so that $q \ge 1$ in Assumptions \ref{a:1} leads to convergence
of the numerical method with order $1$.

Had we carried out the above analysis in the case $\Lambda=0$ we would
have obtained the probabilistic Euler method \eqref{eq:prob2},
and hence \eqref{eq:prob} at grid points, used as our canonical
example in the earlier developments.

\subsection{Backward error analysis}

The idea of backward error analysis is to identify a modified equation 
which is solved by the numerical method either exactly, or at least 
to a higher degree of accuracy than 
the numerical method solves the original equation. 
In the context of differential equations 
this modified equation will involve the step-size $h$. In the setting of 
ordinary differential equations, and the random integrators introduced in this 
section, we will show that the modified equation is a stochastic differential
equation (SDE) in which only the matrix $Q$ from Assumption \ref{a:1} enters;
the details of the random processes used in our construction do not enter
the modified equation. This universality property underpins the methodology
we introduce as it shows that many different choices of random processes all
lead to the same effective behaviour of the numerical method.

We introduce the operators $\oL$ and $\hL$ defined so that, for all
$\phi \in C^{\infty}(\bbR^n,\bbR)$,
\begin{equation}
\label{eq:q1}
\phi\bigl(\Phi_h(u)\bigr)=\bigl(e^{h\oL}\phi\bigr)(u), \quad
\IE \phi\bigl(U_1|U_0=u\bigr)=\bigl(e^{h\hL}\phi\bigr)(u).
\end{equation}
Thus $\oL:=f \cdot \nabla$ and $e^{h\hL}$ is the kernel for
the Markov chain generated by the probabilistic integrator \eqref{eq:NM}. 
In fact we never need to work with $\hL$ itself in what follows, only
with $e^{h\hL}$, so that questions involving the operator logarithm do not need
to be discussed.

We now introduce a modified ODE and a modified SDE which will be needed
in the analysis that follows. The modified ODE is
\begin{equation}
\label{eq:q2}
\frac{d\hat{u}}{dt}=f^h(\hat{u})
\end{equation}
whilst the modified SDE has the form
\begin{equation}
\label{eq:q3}
d\tilde{u}=f^h(\tilde{u})dt+\sqrt{h^{2p} Q} \, dW.
\end{equation}
The precise choice of $f^h$ is detailed below.
Letting $\IE$ denote expectation with respect to $W$, we
introduce the operators $\aL$ and $\tL$ so that, for all
$\phi \in C^{\infty}(\bbR^n,\bbR)$,
\begin{equation}
\label{eq:q4}
\phi\bigl(\hat{u}(h)|\hat{u}(0)=u\bigr)=\bigl(e^{h\aL}\phi\bigr)(u), \quad
\IE \phi\bigl(\tilde{u}(h)|\tilde{u}(0)=0\bigr)=\bigl(e^{h\tL}\phi\bigr)(u).
\end{equation}
Thus 
\begin{equation}
\label{eq:q5}
\aL:=f^h \cdot \nabla, \quad 
\tL=f^h \cdot \nabla+\frac12h^{2p}Q:\nabla \nabla.
\end{equation}
where $:$ denotes the inner product on
$\mathbb{R}^{n \times n}$ which induces the Frobenius norm, that is, $A:B= \text{trace}(A^{T} B)$.
 
The fact that the deterministic numerical integrator has uniform local 
truncation error of order $q+1$ (Assumption \ref{a:2}) implies that, since  $\phi \in C^{\infty}$, 
\begin{equation} \label{eq:local_error}
e^{h \oL}\phi(u)-\phi(\Psi_{h}(u))=\mathcal{O}(h^{q+1}).
\end{equation}
The theory of modified equations for classical one-step numerical
integration schemes for ODEs \cite{hairer1993solving} establishes that it is possible to 
find $f^h$ in the form
\begin{equation}
\label{eq:q6}
f^h:=f+\sum_{i=q}^{q+l}h^{i}f_{i}
\end{equation} 
\begin{equation} \label{eq:local_error_modified}
e^{h \aL}\phi(u)-\phi(\Psi_{h}(u))=\mathcal{O}(h^{q+2+l}).
\end{equation}
We work with this choice of $f^h$ in what follows.

Now for our stochastic numerical method we have 
\begin{eqnarray*}
\phi(U_{k+1}) &=& \phi(\Psi_{h}(U_{k}))+\xi_{k}(h) \cdot \nabla \phi(\Psi_{h}(U_{k}))\\&+&\frac{1}{2} \xi_{k}(h)\xi^{T}_{k}(h) : \nabla \nabla \phi(\Psi_{h}(U_{k}))+ {\mathcal O}(|\xi_k(h)|^3).
\end{eqnarray*}
Furthermore the last term has mean of size 
${\mathcal O}(|\xi_k(h)|^4)$. From Assumption~\ref{a:1} we know 
that $\E \left( \xi_{k}(h)\xi^{T}_{k}(h) \right) = Q h^{2p+1}.$
Thus
\begin{equation}
\label{eq:q7}
e^{h\hL}\phi(u) - \phi\bigl(\Psi_{h}(u)\bigr)= \frac{1}{2} h^{2p+1} Q  : \nabla \nabla \phi\bigl(\Psi_{h}(u)\bigr)+\mathcal{O}(h^{4p+2}).
\end{equation}
From this it follows that
\begin{equation}
\label{eq:one_step1}
e^{h\hL}\phi(u) - \phi\bigl(\Psi_{h}(u)\bigr)= \frac{1}{2} h^{2p+1} Q  : \nabla \nabla \phi(u)+\mathcal{O}(h^{2p+2}).
\end{equation}
Finally we note that \eqref{eq:q5} implies that
\begin{align*}
e^{h\tL}\phi(u)-e^{h\aL}\phi(u)&= e^{h\aL}\bigl(e^{\frac12 h^{2p+1}Q:\nabla\nabla}-I\bigr)\phi(u)\\
&=e^{h\aL}\Bigl(\frac12 h^{2p+1}Q:\nabla\nabla \phi(u)+\mathcal{O}(h^{4p+2})\Bigr)\\
&=\bigl(I+\mathcal{O}(h)\bigr)\Bigl(\frac12 h^{2p+1}Q:\nabla\nabla \phi(u)+\mathcal{O}(h^{4p+2})\Bigr).
\end{align*}
Thus we have
\begin{equation}
\label{eq:q8}
e^{h\tL}\phi(u)-e^{h\aL}\phi(u)=\frac12 h^{2p+1}Q:\nabla\nabla \phi(u)+\mathcal{O}(h^{2p+2}).
\end{equation}

Now using  \eqref{eq:local_error_modified}, \eqref{eq:one_step1} and \eqref{eq:q8} we obtain 
\begin{equation}
\label{eq:q10}
e^{h\tL}\phi(u)-e^{h\hL}\phi(u)=\mathcal{O}(h^{2p+2})+\mathcal{O}(h^{q+2+l}).
\end{equation}
Balancing these terms, in what follows we make the choice $l=2p-q$. 
If $l<0$ we adopt the convention that the drift $f^h$ 
is simply $f.$ With this choice of $q$ we obtain 
\begin{equation}
\label{eq:q9}
e^{h\tL}\phi(u)-e^{h\hL}\phi(u)=\mathcal{O}(h^{2p+2}).
\end{equation}

This demonstrates that the error between the Markov kernel of
one-step of the SDE \eqref{eq:q3} and the Markov kernel
of the numerical method \eqref{eq:NM} is of order $\mathcal{O}(h^{2p+2})$. Some
straightforward stability considerations show that the weak
error over an $\mathcal{O}(1)$ time-interval is $\mathcal{O}(h^{2p+1})$.
We make assumptions giving this stability and then state
a theorem comparing the weak error with respect to the modified
Equation \eqref{eq:q3}, and the original Equation \eqref{eq:ODE}.

\begin{assumption} 
\label{a:3}
The function $f$ is in $C^{\infty}$ and all its derivatives are uniformly
bounded on $\R^n$. Furthermore $f$ is such that the operators $e^{h\oL}$
and $e^{h\hL}$
satisfy, for all $\psi \in C^{\infty}(\R^n,\R)$ and some $L>0$,
\begin{align*}
\sup_{u \in \R^n}|e^{h\oL}\psi(u)| &\le (1+Lh)\sup_{u \in \R^n}|\psi(u)|,\\
\sup_{u \in \R^n}|e^{h\hL}\psi(u)| &\le (1+Lh)\sup_{u \in \R^n}|\psi(u)|.
\end{align*}
\end{assumption}

\begin{remark}
\label{rem:2.3}
If $p=q$ in what follows (our recommended choice) then the weak order of
the method coincides with the strong order; however, measured
relative to the modified equation, the weak order is then
one plus twice the strong order. In this case, the second part of Theorem \ref{t:1} gives us the first weak order result in Theorem \ref{thm:weak}. 
Additionally, Assumption \ref{a:3} is stronger than we need, but
allows us to highlight probabilistic ideas whilst keeping overly 
technical aspects of the numerical analysis to a minimum. 
More sophisticated, but structurally similar, analysis 
would be required for weaker assumptions on $f$. 
Similar considerations apply to the assumptions on $\phi$.
\end{remark}

\begin{theorem} \label{thm:weak}
Consider the numerical method \eqref{eq:probg} and assume that Assumptions  \ref{a:1} and \ref{a:3} are satisfied. 
Then, for $\phi \in C^{\infty}$ function with all derivatives
bounded uniformly on $\R^n$, we have that 
\[
|\phi(u(T))-\E\bigl(\phi(U_{k})\bigr) | \leq K h^{\min\{2p,q\}}, \quad kh=T, 
\]
and
\[
|\IE\bigl(\phi(\tilde{u}(T))\bigr)-\E\bigl(\phi(U_{k})\bigr) | \leq K h^{2p+1}, \quad kh=T, 
\]
where $u$ and $\tilde{u}$ solve \eqref{eq:ODE} and 
\eqref{eq:q3} respectively.
\end{theorem}

\begin{proof}
We prove the second bound first. Let $w_k=\IE\bigl(\phi(\tilde{u}(t_k))|
\tilde{u}(0)=u\bigr)$ and $W_k=\E\bigl(\phi(U_{k})|U_0=u).$
Then let $\delta_k=\sup_{u \in \R^n}|W_k-w_k|.$ It follows from the Markov property that
\begin{align*}
W_{k+1}-w_{k+1} & =e^{h\hL}W_k-e^{h\tL}w_k\\
&=e^{h\hL}W_k-e^{h\hL}w_k+\bigl(e^{h\hL}w_k-e^{h\tL}w_k).
\end{align*}
Using \eqref{eq:q9} and Assumption \ref{a:3} we obtain
$$\delta_{k+1} \le (1+Lh)\delta_k+\mathcal{O}(h^{2p+2}).$$
Iterating and employing the Gronwall inequality gives the
second error bound.

Now we turn to the first error bound, comparing with the solution $u$
of the original equation \eqref{eq:ODE}. From \eqref{eq:one_step1} and
then \eqref{eq:local_error} we see that 
\begin{align*} 
&e^{h \hL}\phi(u)-\phi(\Psi_{h}(u))=\mathcal{O}(h^{2p+1}),\\
&e^{h\oL}\phi(u)-e^{h\hL}\phi(u)=\mathcal{O}(h^{\min\{2p+1,q+1\}}).
\end{align*}
This gives the first weak error estimate, 
after using the stability estimate on $e^{h\oL}$ 
from Assumption \ref{a:3}.
\end{proof}

\begin{example}
Consider the probabilistic integrator derived from the Euler method in 
dimension $n=1$. We thus have  $q=1$, and we hence set $p=1$.
The results in \cite{hairer1} allow us to calculate $f^h$ with $l=1$. 
The preceding theory then leads to strong order of convergence $1$, 
measured relative to the true ODE \eqref{eq:ODE}, and weak order $3$ relative 
to the SDE
\begin{equation*} \label{eq:mod_Euler}
d\hat{u}=\left(f(\hat{u})-\frac{h}{2}f'(\hat{u})f(\hat{u})+\frac{h^{2}}{12}\left(f''(\hat{u})f^{2}(\hat{u})+4 (f'(\hat{u}))^{2}f(\hat{u})\right)\right)dt+\sqrt{C}hdW.
\end{equation*}
\end{example}

With these results now available, the following section provides an empirical study of our probabilistic integrators.


\section{Statistical Inference and Numerics}
\label{sec:numerics}
This section explores applications of the randomised ODE solvers developed in Section \ref{sec:ode}. First, we study forward uncertainty propagation and propose a method for calibrating the problem-dependent scaling constant, $\sigma$, against classic error indicators.\footnote{Recall 
that throughout we assume that, within the context of Assumption \ref{a:1}, 
$Q=\sigma I$. More generally it is possible to calibrate an arbitrary
positive semi-definite $Q$.} Secondly, we employ the calibrated 
measure within
Bayesian inference problems, demonstrating that the resulting posteriors exhibit more consistent behaviour over varying step-sizes than with naive use of a deterministic integrators.
Throughout this section we use the FitzHugh--Nagumo model to illustrate ideas \cite{ramsay2007parameter}.
This is a two-state non-linear oscillator, with states $(V,R)$ and parameters $(a,b,c)$, governed by the equations
\begin{equation}
\frac{dV}{dt} = c \left(V - \frac{V^3}{3} +R \right),\;\;\;
\frac{dR}{dt}  = -\frac{1}{c}\left(V-a+bR \right).
\label{eq:FH}
\end{equation}
This particular example does not satisfy the stringent Assumptions \ref{a:2}
and \ref{a:3} and the numerical results shown demonstrate that, as indicated
in Remarks \ref{rem:2.1} \and \ref{rem:2.3}, our theory will extend to
weaker assumptions on $f$, something we will address in future work.

\begin{figure}[htbp]
\centering
\includegraphics[scale=.7]{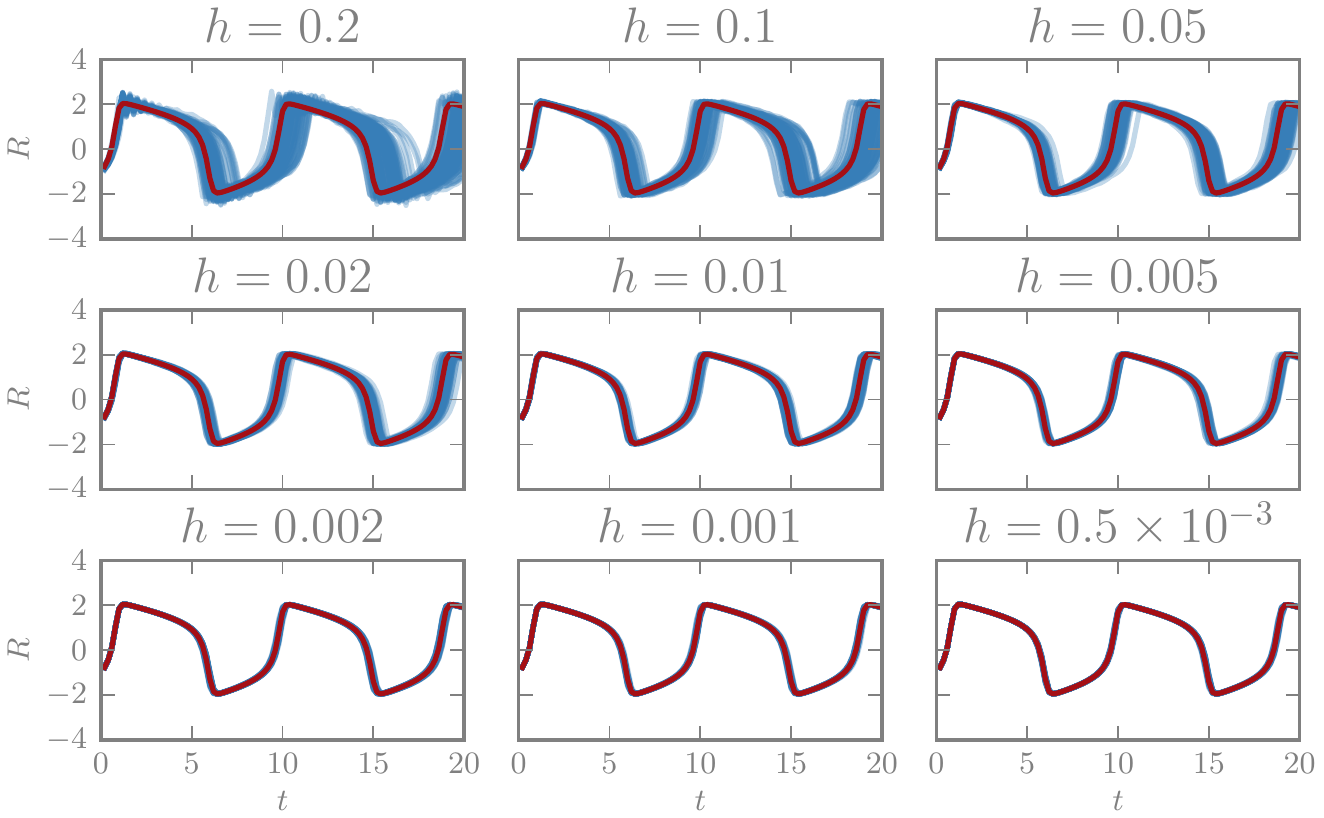}

\caption{The true trajectory of the $V$ species of the FitzHugh--Nagumo model (red) and one hundred realisations from a probabilistic Euler ODE solver with various step-sizes and noise scale $\sigma=.1$ (blue).}
\label{fig:fitz_steps}
\end{figure}

\begin{figure}[htb]
\centering
\includegraphics[scale=.7]{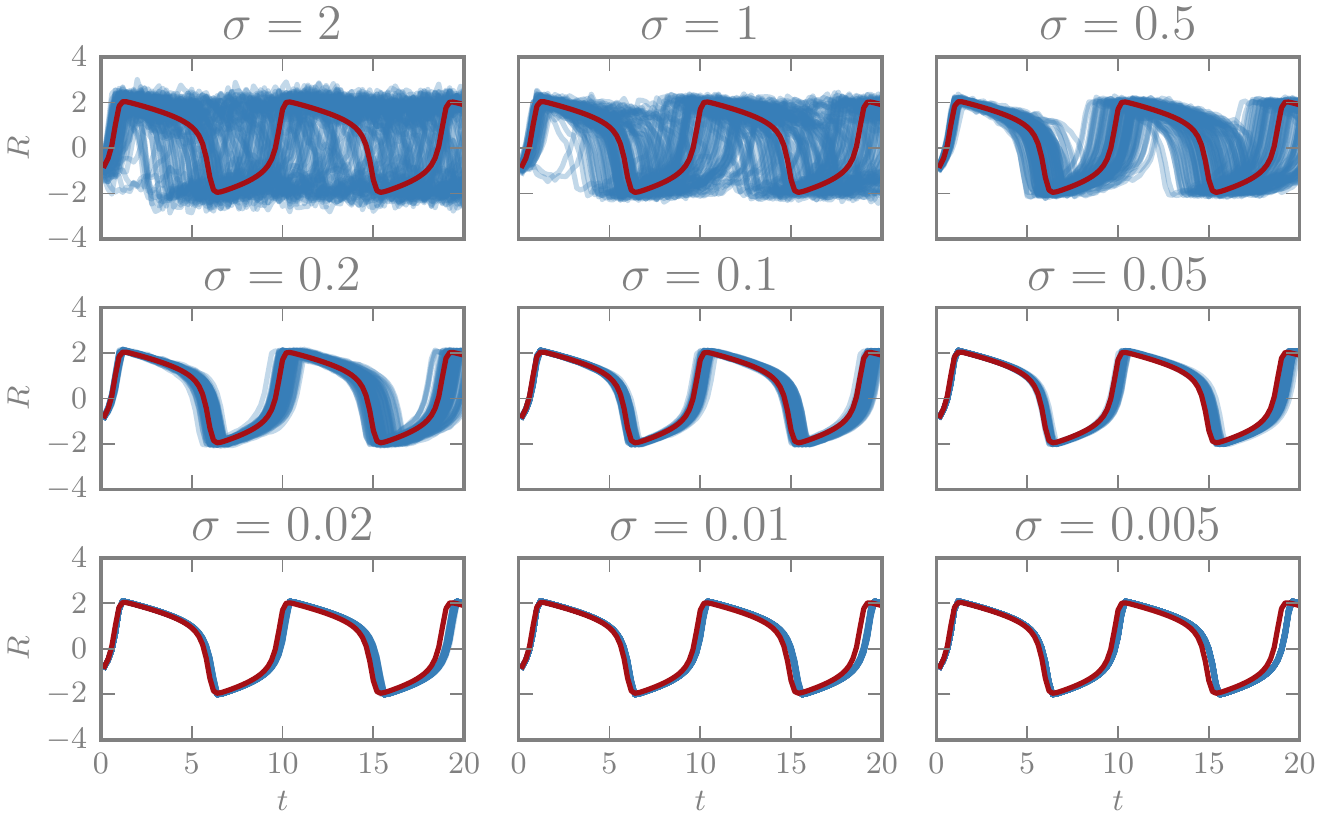}

\caption{The true trajectory of the $V$ species of the FitzHugh--Nagumo model (red) and one hundred random draws from a probabilistic Euler ODE solver with step-size $h=.1$ and various noise scalings (blue). Observe that for small $\sigma$ the draws concentrate away from the true solution. The value of $\sigma = 0.2$ appears appropriate, since the true value lies comfortably within the solution draws. }
\label{fig:fitz_noise}
\end{figure}

\subsection{Calibrating forward uncertainty propagation}
\label{sec:numerics:calibration}
Consider Equation \eqref{eq:FH} with fixed initial conditions $V(0)=-1, R(0)=1$, and parameter values $(.2,.2,3)$. Figure \ref{fig:fitz_steps} shows draws of the $V$ species trajectories from the measure associated with the probabilistic Euler solver with $p=q=1$,  for various values of the step-size and fixed $\sigma=0.1$. The contraction of the measure towards the reference solution, as $h$ shrinks, is clearly evident. Furthermore, the uncertainty exhibits interesting, non-Gaussian structure where trajectories disagree whether to begin the steep phase of the oscillation.

Although the rate of contraction is governed by the underlying deterministic method, the scale parameter, $\sigma$, completely controls the apparent uncertainty in the solver. To illustrate this, Figure \ref{fig:fitz_noise} fixes the step-size, $h=0.1$, and shows that rescaling the noise can create any apparent level of certainty desired, including high confidence in an incorrect solution. This tuning problem exists in general, since $\sigma$ is problem dependent and cannot obviously be computed analytically.

Therefore, we propose to calibrate $\sigma$ to replicate the amount of error suggested by classical error indicators. In the following discussion, we often explicitly denote the dependence on $h$ and $\sigma$ with superscripts, hence the probabilistic solver is $\Uhs$ and the corresponding deterministic solver is $\Uhz$.  Define the true error as $e(t) = u(t) - \Uhz(t)$. Then we assume there is some computable error indicator $E(t) \approx e(t)$, defining $E_k = E(t_k)$. The simplest error indicators might compare differing step-sizes, $E(t) = \Uhz(t) - U^{2h, 0}(t)$, or differing order methods, as in a Runge--Kutta 4-5 scheme.

We proceed by constructing a probability distribution $\pi(\sigma)$ that is
maximised when the overall amount of error produced by our probabilistic integrator
matches that suggested by the error indicator. We perform this scale matching 
by: (i)  using a Gaussian approximation of our random solver at each step $k$,
$
\tilde{\mu}_k^{h,\sigma} = \mathcal{N}(\mathbb{E} (\Uhs_k), \mathbb{V}(\Uhs_k ));
$
and (ii) by constructing a second Gaussian measure from the deterministic solver, $\Uhz_k$, and the available error indicator, $E_k$,
$
\nu_k^\sigma = \mathcal{N}(\Uhz_k, (E_k)^2).
$
Thus (ii) is a natural Gaussian representation of the uncertainty based upon information from deterministic methods. We construct $\pi(\sigma)$ by penalising the distance between these two normal distributions at every step:
$
\pi(\sigma) \propto \prod_k \exp \left( -d(\tilde{\mu}_k^{h,\sigma}, \nu_k^\sigma)  \right)
\label{eq:sigmaDist}
$.
We find that the Bhattacharyya distance (closely related to the Hellinger metric) works well \cite{Kailath1967bhattacharyya}, 
since it diverges quickly if either the mean or variance differs between the two inputs. Computing $\pi(\sigma)$ requires the mean and variance of the process $\Uhs_k$, which can be estimated using Monte Carlo. If the ODE state is a vector, we take the product of the univariate Bhattacharyya distances. Note that this calibration depends not only on the problem of interest, but also on the initial conditions and any parameters of the ODE.

Returning to the FitzHugh--Nagumo model, we apply this strategy for choosing the noise scale. We sampled from $\pi(\sigma)$ using pseudomarginal MCMC \cite{andrieu2009pseudo}. The sampled densities are shown in Figure \ref{fig:fitz_fwd_density}; the similarity in the inferred $\sigma$ suggests that the asymptotic convergence rate is nearly exact in this case. Since this distribution is well peaked, MCMC mixes very quickly and is well represented by the MAP value, hence we proceed using $\sigma^\ast = \argmax \pi(\sigma)$. Next, we examine the quality of the scale matching by plotting the magnitudes of the random variation against the error indicator in Figure \ref{fig:fitz_fwd_summary}, for several different step-sizes, observing good agreement of the marginal variances. 
Although the marginal variances are empirically matched to the error indicator, our measure is still reveals non-Gaussian structure and correlations in time not revealed by the deterministic analysis. 
As described, this procedure requires fixed inputs to the ODE, but it is straightforward to marginalise out a prior distribution over input parameters.

\begin{figure}[htb]
\centering
\includegraphics[scale=.5]{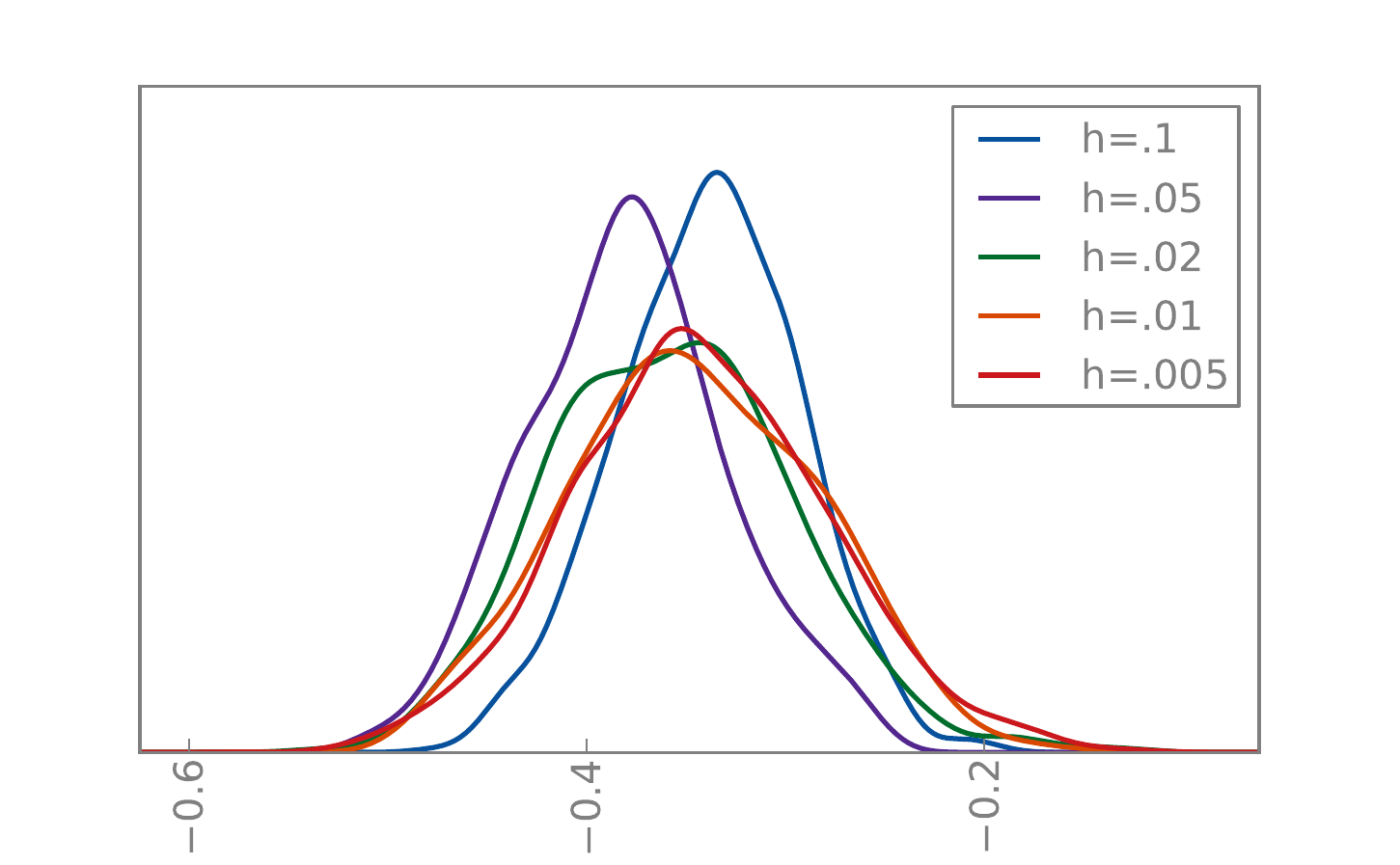}

\caption{The sampled density of $\pi(\sigma)$ in the FitzHugh--Nagumo model, for several different step-sizes.}
\label{fig:fitz_fwd_density}
\end{figure}

\begin{figure}[htb]
\centering
\includegraphics[scale=.5]{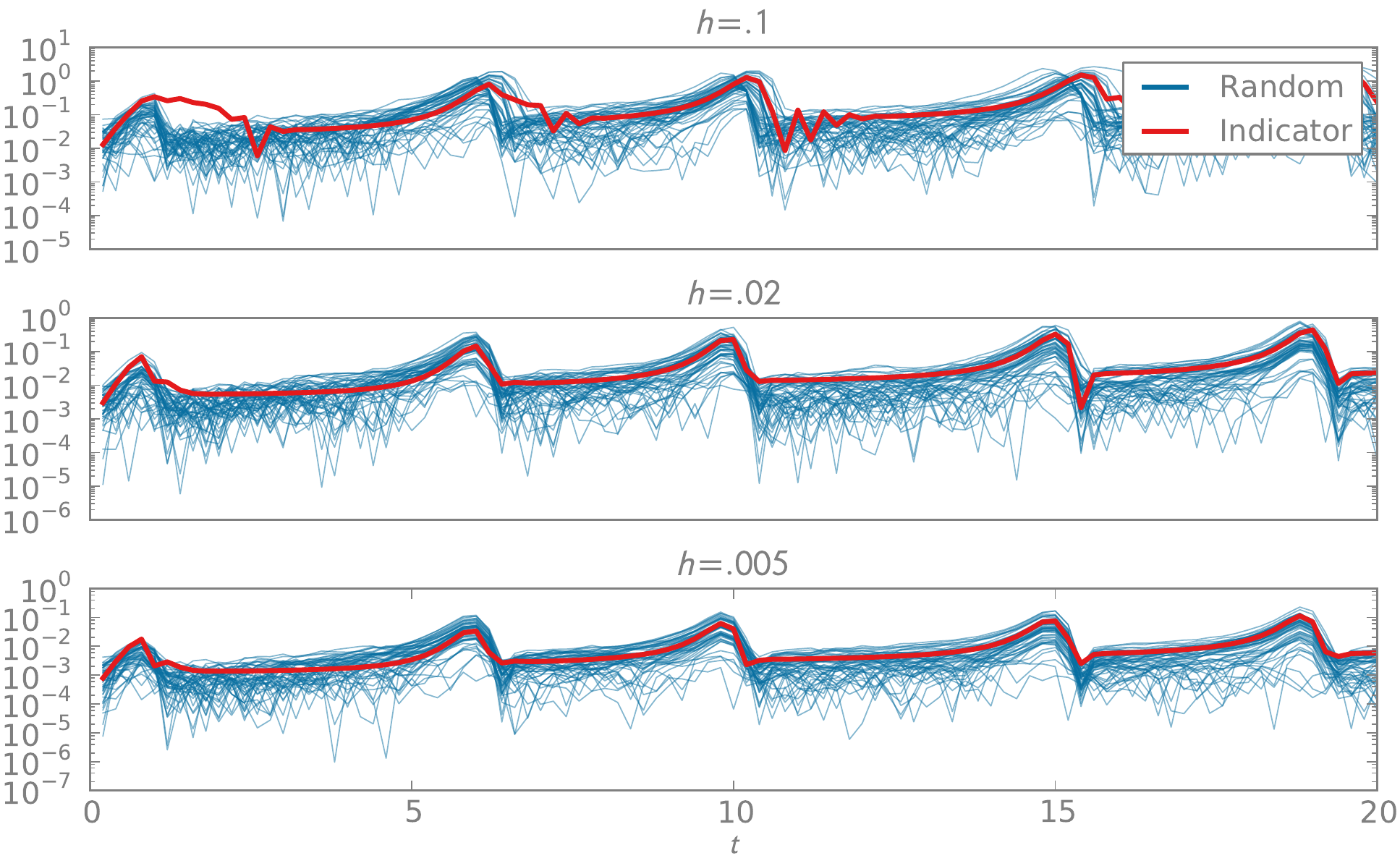}

\caption{A comparison of the error indicator for the $V$ species of the FitzHugh--Nagumo model (blue) and the observed variation in the calibrated probabilistic solver. The red curves depict fifty samples of the magnitude of the difference between a standard Euler solver for several step-sizes and the equivalent randomised variant, using $\sigma^\ast$, maximising $\pi(\sigma)$.}
\label{fig:fitz_fwd_summary}
\end{figure}

%
%

\subsection{Constructing Bayesian posterior inference problems}

Given the calibrated probabilistic ODE solvers described above, let us consider how to incorporate them into inference problems.

Assume we are interested in inferring parameters of the ODE given noisy observations of the state. Specifically, we wish to infer parameters $\theta \in \mathbb{R}^d$ for the differential equation $\dot{u} = f(u, \theta)$, with fixed initial conditions $u(t=0) = u_0 $ (a straightforward modification may include inference on initial conditions). Assume we are provided with data $d \in \mathbb{R}^m$,  $d_j = u(\tau_j) + \eta_j$ at some collection of times $\tau_j$, corrupted by i.i.d. noise, $\eta_j \sim \mathcal{N}(0, \Gamma)$. If we have prior $\mathbb{Q}(\theta)$, the posterior we wish to explore is,
$
\mathbb{P}(\theta \mid d) \propto \mathbb{Q}(\theta) \mathcal{L}(d,u(\theta)),
$
where density $\mathcal{L}$ compactly summarises this likelihood model. 

The standard computational strategy is to simply replace the unavailable trajectory $u$ with a numerical approximation, inducing approximate posterior 
$
\Phz(\theta \mid d) \propto \mathbb{Q}(\theta) \mathcal{L}(d,\Uhz(\theta)).
$
Informally, this approximation will be accurate when the error in the numerical solver is small compared to $\Gamma$ and often converges formally to $\mathbb{P}(\theta \mid d)$ as $h \to 0$ \cite{DS15}. However, the error in $\Uhz$ might be non-Gaussian and highly correlated, which could make it surprisingly easy for errors at finite $h$ to have substantial impact.

In this work, we are concerned about the undue optimism in the predicted variance; without a rigorous treatment of the solver uncertainty, the posterior can concentrate around an arbitrary parameter value even when the deterministic solver is inaccurate, and is merely able to reproduce the data by coincidence.
The more conventional concern is that any error in the solver will be transferred into posterior bias.
Practitioners commonly alleviate both concerns by tuning solvers until they are nearly perfect, however, we note that this may be computationally prohibitive in many contemporary statistical applications.


We can construct a different posterior that includes the uncertainty in the solver by taking an expectation over random solutions to the ODE
\begin{equation}
\Phs(\theta \mid d) \propto \mathbb{Q}(\theta)  \int \mathcal{L}(d,\Uhs(\theta, \xi)) d\xi.
\end{equation}
Intuitively, this construction favours parameters that exhibit agreement with 
the entire family of uncertain trajectories, not just the one produced by the deterministic integrator. The typical effect of this expectation is to increase the posterior uncertainty on $\theta$, preventing the inappropriate posterior collapse we are concerned about. Indeed, if the integrator cannot resolve the underlying dynamics, $h^{p+1/2}\sigma$ will be large. Then $\Uhs(\theta, \xi)$ is independent of $\theta$, hence the prior is recovered, $\Phs(\theta \mid d) \approx \mathbb{Q}(\theta)$.

Notice that as $h \to 0$, both the measures $\Phz$ and $\Phs$ collapse to the posterior using the analytic solution, $\mathbb{P}$, hence both methods are correct. We do not expect the bias of $\Phs$ to be improved, since all of the averaged trajectories are of the same quality as the deterministic solver in $\Phz$. We now construct an analytic inference problem demonstrating these behaviours.

\begin{example}
Consider inferring the initial condition, $u_0\in \mathbb{R}$,  of the 
scalar linear differential equation, 
$\dot{u} = \lambda u,$ with $\lambda>0.$ 
We apply a numerical method to produce
the approximation $U_k \approx u(kh)$.
We observe the state at some times $t=kh$, with additive
noise $\eta_k \sim \mathcal{N}(0, \gamma^2)$: $d_k = U_k + \eta_k$. 
If we use a deterministic Euler solver, the model predicts
$U_{k} = (1+h\lambda)^k u_0.$
These model predictions coincide with the slightly perturbed problem 
$$\frac{du}{dt} = h^{-1}\log(1+\lambda h)u,$$ 
hence error increases with time. However, the assumed observational model does not allow for this, as the observation variance is $\gamma^2$ 
at all times. 

In contrast, our proposed probabilistic Euler solver predicts 
$$U_{k} = (1+h\lambda)^k u_0 + \sigma h^{3/2} \sum_{j=0}^{k-1} \xi_j (1+\lambda h)^{k-j-1},$$
where we have made the natural choice $p=q$, where $\sigma$ is the problem 
dependent scaling of the noise and the $\xi_k$ are i.i.d. $\mathcal{N}(0,1).$ 
For a single observation, $\eta_k$ and every $\xi_k$ are independent, so we may rearrange the equation to consider the perturbation as part of the observation operator. Hence, a single observation at $k$ has effective variance
$$\gamma_h^2:=\gamma^2 +  \sigma^2 h^{3} \sum_{j=0}^{k-1} (1+\lambda h)^{2(k-j-1)}=
\gamma^2+\sigma^2 h^{3}\frac{(1+\lambda h)^{2k}-1}{(1+\lambda h)^2-1}.$$
Thus, late time observations are modelled as being increasingly inaccurate.

Consider inferring $u_0$, given a single observation $d_k$
at time $k$. If a Gaussian prior $\mathcal{N}(m_0,\zeta_0^2)$
is specified for $u_0$, then the posterior is $\mathcal{N}(m,\zeta^2)$,
where
\begin{align*}
\zeta^{-2}=\frac{(1+h\lambda)^{2k}}{\gamma_h^2}+\zeta_0^{-2}, &\qquad 
\zeta^{-2}m=\frac{(1+h\lambda)^k d_k}{\gamma_h^2}+\zeta_0^{-2}m_0.
\end{align*}
The observation precision is scaled by $(1+h\lambda)^{2k}$ because late time data contain increasing information.
Assume that the data are $d_k=e^{\lambda kh}u_0^{\dagger}+\gamma \eta^{\dagger}$,
for some given true initial condition $u_0^{\dagger}$ and noise realisation
$\eta^{\dagger}.$ Consider now the asymptotic regime where $h$ is fixed and
$k \to \infty$. For the standard Euler method, where $\gamma_h=\gamma$,
we see that $\zeta^2 \to 0$, whilst $m \asymp \bigl((1+h\lambda)^{-1}e^{h\lambda}\bigr)^{k} u_0^\dagger$. Thus the inference scheme becomes increasingly certain of
the wrong answer: the variance tends to zero and the mean tends to infinity.

In contrast, with a randomised
integrator, the fixed $h$, large $k$ asymptotics are
\begin{align*}
\zeta^2 \asymp \frac{1}{\zeta_0^{-2}+\lambda(2+\lambda h)\sigma^{-2}h^{-2}}, \qquad 
m \asymp \frac{\left((1+h\lambda)^{-1}e^{h\lambda}\right)^{k} u_0^\dagger }{1+\zeta_0^{-2}\sigma^2h^2\lambda^{-1}(2+\lambda h)^{-1}}.
\end{align*}
Thus, the mean blows up at a modified rate, but the variance remains positive.
%
%
%
\end{example}

We take an empirical Bayes approach to choosing $\sigma$, that is, using a constant, fixed value $\sigma^\ast = \argmax\pi(\sigma)$, chosen before the data is observed. Joint inference of the parameters and the noise scale suffer from well-known MCMC mixing issues in Bayesian hierarchic models. 

We now return to the FitzHugh--Nagumo model; given fixed initial conditions, we attempt to recover parameters $\theta = (a,b,c)$ from observations of both species at times $\tau = 1,2,\ldots,40$. The priors are  log-normal, centred on the true value with unit variance, and with observational noise $\Gamma = 0.001$. 
The data is generated from a high quality solution, and we perform inference using Euler integrators with various step-sizes, $h \in \{0.005, 0.01, 0.02, 0.05, 0.1\}$, spanning a range of accurate and inaccurate integrators.

We first perform the inferences with naive use of deterministic Euler integrators. We simulate from each posterior using delayed rejection MCMC \cite{haario2006dram} for 100,000 steps, discarding the first 10\% as burn in. The posteriors are shown in Figure \ref{fig:fitz_detPost}. Observe the undesirable
 concentration of every posterior, even those with poor solvers; the posteriors are almost mutually singular, hence clearly the posterior widths are meaningless.
 
Secondly, we repeat the experiment using our probabilistic Euler integrators, with results shown in Figure \ref{fig:fitz_randPost}. We use a noisy pseudomarginal MCMC method, whose fast mixing is helpful for these initial experiments \cite{medina2015noisymh}. These posteriors are significantly improved, exhibiting greater mutual agreement and obvious increasing concentration with improving solver quality. The posteriors are not perfectly nested, possible evidence that our choice of scale parameter is imperfect, or that the assumption of locally Gaussian error deteriorates. As expected, the bias of these posteriors is essentially unchanged.

%

\begin{figure}[p]
\centering
\includegraphics[scale=.6,trim=0in .25in 0in .25in]{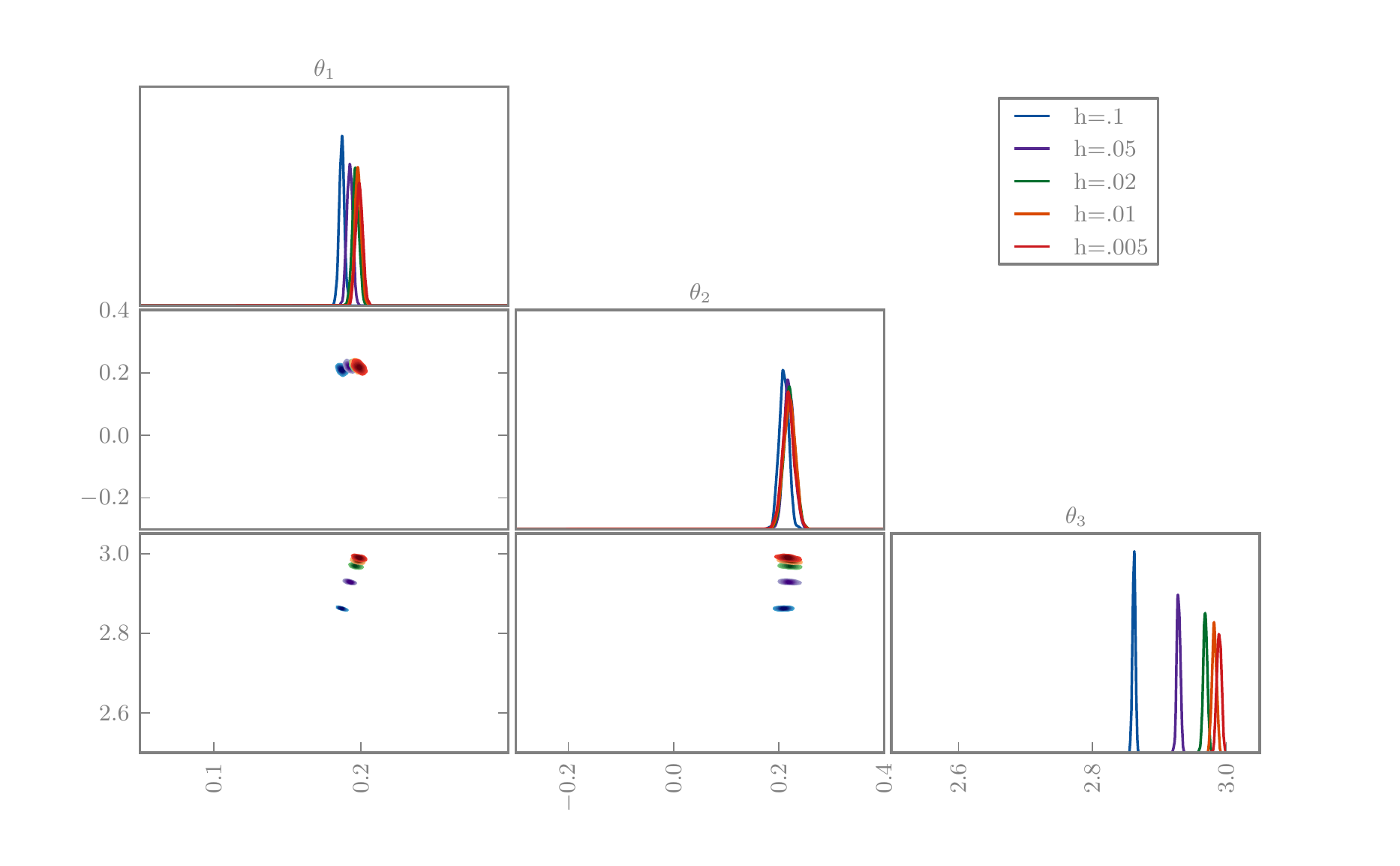}

\caption{The posterior marginals of the FitzHugh--Nagumo inference problem using deterministic integrators with various step-sizes.}
\label{fig:fitz_detPost}
\end{figure}

\begin{figure}[p]
\centering
\includegraphics[scale=.6,trim=0in .25in 0in .25in]{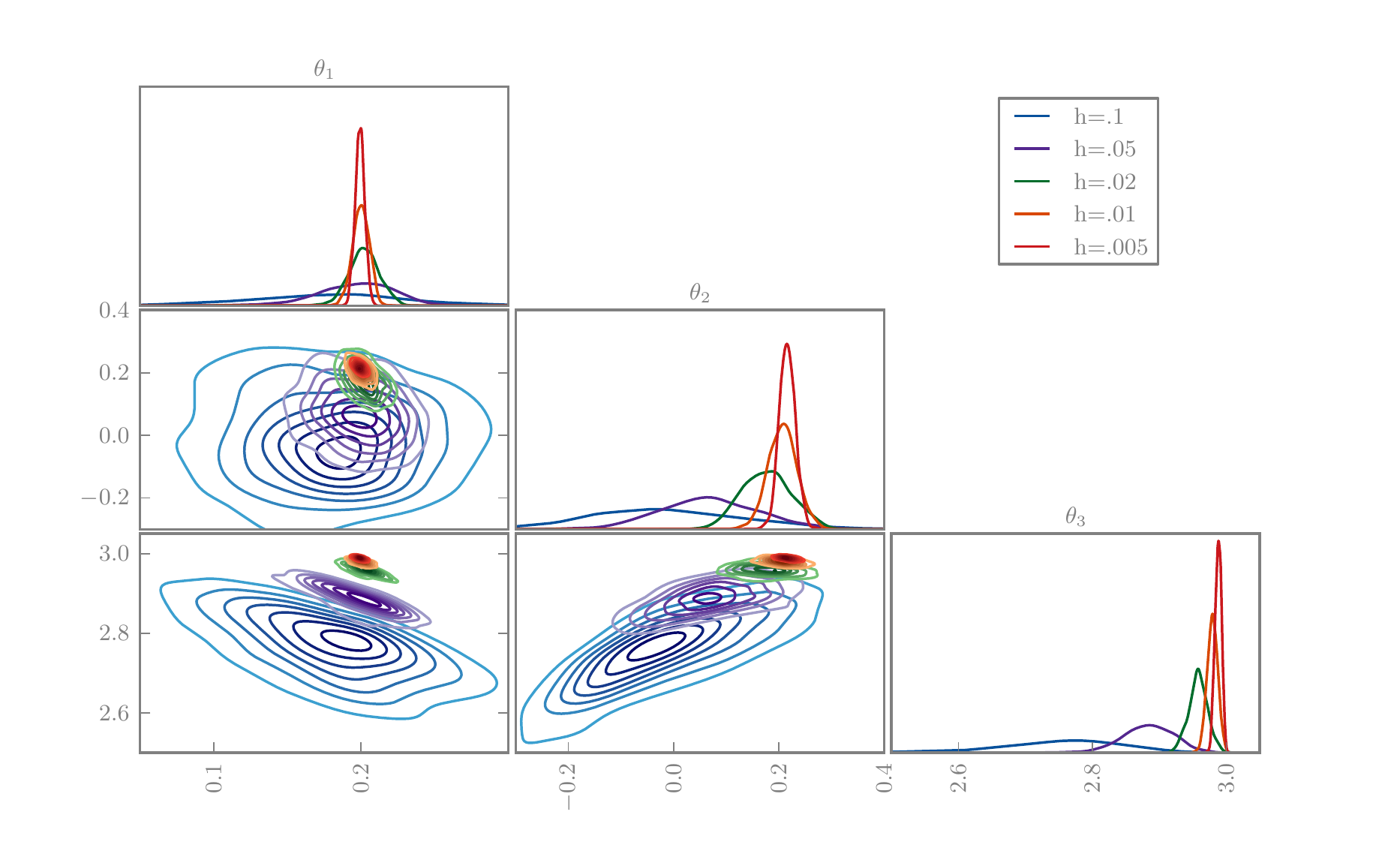}

\caption{The posterior marginals of the FitzHugh--Nagumo inference problem using probabilistic integrators with various step-sizes.}
\label{fig:fitz_randPost}
\end{figure}

\section{Probabilistic Solvers for Partial Differential Equations}
\label{sec:pde}

We now turn to present a framework for probabilistic solutions to partial differential equations, working within the finite element setting.
Our approach to PDE problems is analogous to that for ODEs, except that now we 
randomly perturb the finite element basis functions. 
The assumptions required, and the strong convergence results, 
closely resemble their ODE counterparts.

\subsection{Probabilistic finite element method for variational problems}

Let $\mathcal{V}$ be a Hilbert space of real valued functions
defined on a bounded polygonal domain $D \subset \bbR^d$.
Consider a weak formulation of a
linear PDE specified via a symmetric bilinear form 
$a:\mathcal{V}\times \mathcal{V}\longrightarrow \mathbb{R}$,
and a linear form $r:\mathcal{V}\longrightarrow \mathbb{R}$ to give the problem
of finding
$
u\in \mathcal{V}: 
\, a(u,v)=r(v), \quad \forall v\in \mathcal{V}.
\label{eq:eq}
$
This problem can be approximated by specifying a finite dimensional subspace
$\mathcal{V}^h \subset \mathcal{V}$  and seeking a solution in $\cV^h$ instead. This leads to a finite dimensional problem to be solved 
for the approximation $U$:
\begin{equation}
U\in \mathcal{V}^h: 
a(U,v)=r(v), \quad \forall v\in \mathcal{V}^h.
\label{eq:ga}
\end{equation}
This is known as the Galerkin method.

We will work in the setting of finite element methods, 
assuming that $\mathcal{V}^h=\mathrm{span}\{\phi_j\}_{j=1}^J$
where $\phi_j$ is locally supported on a grid of points $\{x_j\}_{j=1}^J.$
The parameter $h$ is introduced to measure the diameter of the finite
elements.  We will also assume that
\begin{equation}
\label{eq:nodal}
\phi_j(x_k)=\delta_{jk}.
\end{equation}
Any element $U \in \cV^h$ can then be written as
\begin{equation}
\label{eq:nodalw}
U(x)=\sum_{j=1}^{J} U_j \phi_j(x)
\end{equation}
from which it follows that $U(x_k)=U_k.$
The Galerkin method then gives the following equation for
$\us=(U_1,\ldots, U_J)^T:$
\begin{equation}
\label{eq:ls}
A\us=\rs
\end{equation}
where
$A_{jk}=a(\phi_j,\phi_k)$ and $\rs_k=r(\phi_k).$ 

In order to account for uncertainty introduced by the
numerical method, 
we will assume that each  basis
function $\phi_j$ can be split into the sum of a systematic
part $\phs_j$ and random part $\phr_j$, where both $\phi_j$ and $\phs_j$ satisfy the
nodal property \eqref{eq:nodal}, hence $\phr_j(x_k) = 0$. Furthermore we assume that
each $\phr_j$ shares the same compact support as
the corresponding $\phs_j$, preserving the sparsity structure of the underlying deterministic method. 
%
%

\subsection{Strong convergence result}
As in the ODE case, we begin our convergence analysis with assumptions constraining the random perturbations and the underlying deterministic approximation.
The bilinear form $a(\cdot,\cdot)$ is assumed to induce an inner-product,
and then norm via $\|\cdot\|_a^2=a(\cdot,\cdot);$ furthermore we
assume that this norm is equivalent to the norm on $\cV$.
Throughout $\E$ denotes expectation with respect to the random basis
functions.

\begin{assumption}
\label{assumption:pde1}
The collection of random basis functions $\{\phr_j\}_{j=1}^J$ are independent,
zero-mean, Gaussian random fields, each of which satisfies $\phr_j(x_k)=0$ and shares the same support as the corresponding systematic basis function $\phs_j.$ 
For all $j$, the number of basis
functions with index $k$ which share the support of the basis functions
with index $j$ is bounded independently of $J$, the total number of
basis functions. 
Furthermore the basis functions are scaled so that
$\sum_{j=1}^J \E \|\phr_j\|_a^2 \le Ch^{2p}.$ 
\end{assumption}

\begin{assumption}
\label{assumption:pde2}
The true solution $u$ of problem \eqref{eq:eq} is in $L^{\infty}(D).$
Furthermore the standard deterministic
interpolant of the true solution, defined by
$\vs:=\sum_{j=1}^J u(x_j)\phs_j,$
satisfies $\|u-\vs\|_a \le Ch^q.$
\end{assumption}

\begin{theorem}
\label{t:p1}
Under Assumptions \ref{assumption:pde1} and \ref{assumption:pde2} it follows that
the approximation $U$, given by \eqref{eq:ga}, satisfies
$$\E \|u-U\|_a^2 \le Ch^{2\min\{p,q\}}.$$
\end{theorem}

As for ODEs, the solver accuracy is limited by either the amount of noise injected or the convergence rate of the underlying deterministic method; the choice
$p=q$ is natural since it introduces the maximum amount of noise that does not 
disturb the deterministic rate of convergence.

\bp
Recall the Galerkin orthogonality property
which follows from subtracting the approximate variational principle
from the true variational principle: it states that, for $e=u-U$,
\begin{equation}
\label{eq:Gorthog}
a(e,v)= 0, \quad \forall v\in \mathcal{V}^h.
\end{equation}
From this it follows that
\begin{equation}
\|e\|_{a}  \le  \|u-v\|_{a}, \quad \forall v \in \cV^h. 
\label{eq:fem3}
\end{equation}
To see this note that,
for any $v \in \cV^h$, 
the orthogonality property \eqref{eq:Gorthog} gives
\begin{eqnarray}
a(e,e) &=& a(e,e+U-v)
= a(e,u-v).
\label{eq:fem1}
\end{eqnarray}
Thus, by Cauchy--Schwarz, $\|e\|_{a}^2  \le  \|e\|_{a}\|u-v\|_{a}, \quad \forall v \in \cV^h $
implying \eqref{eq:fem3}. We now set, for $v \in \cV$, 
\begin{eqnarray*}
v(x) &=& \sum_{j=1}^J u(x_j)\phi_j(x)
= \sum_{j=1}^J u(x_j)\phs_j(x)+\sum_{j=1}^J u(x_j)\phr_j(x)\\
&=:& \vs(x)+\vr(x).
\end{eqnarray*}
By the mean-zero and independence properties of the random basis functions 
we deduce that
\begin{eqnarray*}
\E \|u-v\|_a^2 &=& \E a(u-v,u-v)
= \E a(u-\vs,u-\vs)+\E a(\vr,\vr)\\
&=& \|u-\vs\|_a^2+\sum_{j=1}^J u(x_j)^2 \E \|\phr_j\|_a^2.
\end{eqnarray*}
The result follows from Assumptions \ref{assumption:pde1} and \ref{assumption:pde2}.
\ep

\subsection{Poisson solver in two dimensions}
Consider a Poisson equation with Dirichlet boundary conditions in dimension $d=2$, namely
\begin{eqnarray*}
&-\triangle u=f, \quad & x \in D,\\
&u=0, \quad & x \in \partial D.
\end{eqnarray*}
We set $\cV=H^1_0(D)$ and $H$ to be the space $L^2(D)$ with
inner-product $\langle \cdot, \cdot \rangle$ and resulting
norm $| \cdot |^2=\langle \cdot, \cdot \rangle.$ 
The weak formulation of the problem has the form \eqref{eq:eq}
with 
$$a(u,v)=\int_{D} \nabla u(x) \nabla v(x) dx, \quad r(v)=\langle f,v \rangle.$$
Now consider piecewise linear finite elements satisfying the assumptions
of Section 4.2 in \cite{jo92} and take these to comprise
the set $\{\phs_j\}_{j=1}^J.$ Then $h$ measures the width of
the triangulation of the finite element mesh. Assuming that $f \in H$ it follows that
$u \in H^2(D)$ and that
\begin{equation}
\label{eq:bnd}
\|u-\vs\|_a \le Ch\|u\|_{H^2}.
\end{equation}
Thus $q=1$.
We choose random basis members $\{\phr_j\}_{j=1}^J$ so that
Assumptions \ref{assumption:pde1} hold with $p=1$. Theorem \ref{t:p1}
then shows that, for $e=u-U$, $\E \|e\|_a^2 \le Ch^{2}.$
Note, however, that in the deterministic case we expect an improved
rate of convergence in the function space $H$. We show that such a result 
also holds in our random setting.

Note that, under Assumptions \ref{assumption:pde1}, 
if we introduce $\gamma_{jk}$ that is $1$ when two basis functions
have overlapping support, and $0$ otherwise, then
$\gamma_{jk}$ is symmetric and there is constant $C$, independent of 
$j$ and $J$, such that $\sum_{k=1}^J \gamma_{jk} \le C.$
Now let $\varphi$ solve the equation
$a(\varphi,v)=\langle e,v \rangle, \quad \forall v \in \cV.$
Then $\|\varphi\|_{H^2} \le C|e|.$
We define $\vars$ and $\varr$ in analogy with the definitions of
$\vs$ and $\vr$. Following the usual arguments for application of the
Aubin--Nitsche trick \cite{jo92}, we have
$
|e|^2=a(e,\varphi)=a(e,\varphi-\vars-\varr).
$
Thus
\begin{equation}
|e|^2  \le  \|e\|_a\|\varphi-\vars-\varr\|_a
 \le  \sqrt{2}\|e\|_a\Bigl(\|\varphi-\vars\|_a^2+\|\varr\|_a^2\Bigl)^{\frac12}. \label{eq:dnb2}
\end{equation}
We note that
$\varr(x)=\sum_{j=1}^{J} \varphi(x_j)\phr_j(x)=\|\varphi\|_{H^2}
\sum_{j=1}^{J} a_j \phr_j(x)$
where, by Sobolev embedding ($d=2$ here),
$a_j:=\varphi(x_j)/\|\varphi\|_{H^2}$
satisfies $\max_{1 \le j \le J}|a_j| \le C.$
Note, however, that the $a_j$ are random and correlated with all of
the random basis functions.
Using this, together with \eqref{eq:bnd}, in \eqref{eq:dnb2}, we obtain
$$|e|^2  \le C\|e\|_a\Bigl(h^2+\bigl\|\sum_{j=1}^{J} a_j \phr_j(x)\bigr\|_a^2\Bigr)^{\frac12}\|\varphi\|_{H^2}.$$
We see that
$$|e| \le C\|e\|_a\Bigl(h^2+
\sum_{j=1}^J \sum_{k=1}^J a_ja_k \,a(\phr_j,\phr_k)\Bigr)^{\frac12}.$$
From this and the symmetry of $\gamma_{jk}$, we obtain
\begin{align*}
|e| &\le C\|e\|_a\Bigl(h^2+\sum_{j=1}^J \sum_{k=1}^J \gamma_{jk}
\bigl(\|\phr_j\|_a^2+\|\phr_k\|_a^2\bigr)\Bigr)^{\frac12}\\
&\le C\|e\|_a\Bigl(h^2+2C\sum_{j=1}^J \|\phr_j\|_a^2\Bigr)^{\frac12}.
\end{align*}
Taking expectations, using that $p=q=1$, we find, using Assumptions \ref{assumption:pde1}, that
$\E |e| \le Ch \bigl(\E\bigl\|e\|_a^2\bigr)^{\frac12} \le Ch^2$
as desired.  Thus we recover the extra
order of convergence over the rate $1$ in the $\|\cdot\|_a$ norm
(although the improved rate is in $L^1(\Omega;H)$ whilst the lower rate
of convergence is in $L^2(\Omega;\cV).$)

\section{PDE Inference and Numerics}
\label{sec:pdenumerics}

We now perform numerical experiments using probabilistic solvers for elliptic PDEs. Specifically, we perform inference in a 1D elliptic PDE, 
$
\nabla \cdot (\kappa(x) \nabla u(x)) = 4x
$
for $x \in [0,1]$, given boundary conditions $u(0) = 0, u(1) = 2$. We represent $\log \kappa$ as piecewise constant over ten equal-sized intervals; the first, on $x \in [0,.1)$ is fixed to be one to avoid non-identifiability issues, and the other nine are given a prior $\theta_i = \log \kappa_i \sim \mathcal{N}(0,1)$. Observations of the field $u$ are provided at $x=(0.1, 0.2, \ldots 0.9)$, with i.i.d. Gaussian error, $\mathcal{N}(0, 10^{-5})$; the simulated observations were generated using a fine grid and quadratic finite elements, then perturbed with error from this distribution. 

Again we investigate the posterior produced at various grid sizes, using both 
deterministic and randomised solvers. The randomised basis functions are draws from a Brownian bridge conditioned to be zero at the nodal points, implemented in practice with a truncated Karhunen--Lo\`{e}ve expansion. The covariance operator may be viewed as a fractional Laplacian, as discussed in \cite{lindgren2011explicit}.
%
%
The scaling $\sigma$ is again determined by maximising the distribution described in Section \ref{sec:numerics:calibration}, where the error indicator compares linear to quadratic basis functions, and we marginalise out the prior over the $\kappa_i$ values.


%

The posteriors are depicted in Figures \ref{fig:elliptic1d_det} and \ref{fig:elliptic1d_rand}. As in the ODE examples, the deterministic solvers lead to incompatible posteriors for varying grid sizes. In contrast, the randomised solvers suggest increasing confidence as the grid is refined, as desired. The coarsest grid size uses an obviously inadequate ten elements, but this is only apparent in the randomised posterior.

\begin{figure}[p]
\centering
\includegraphics[scale=.6, trim=0in .25in 0in .25in]{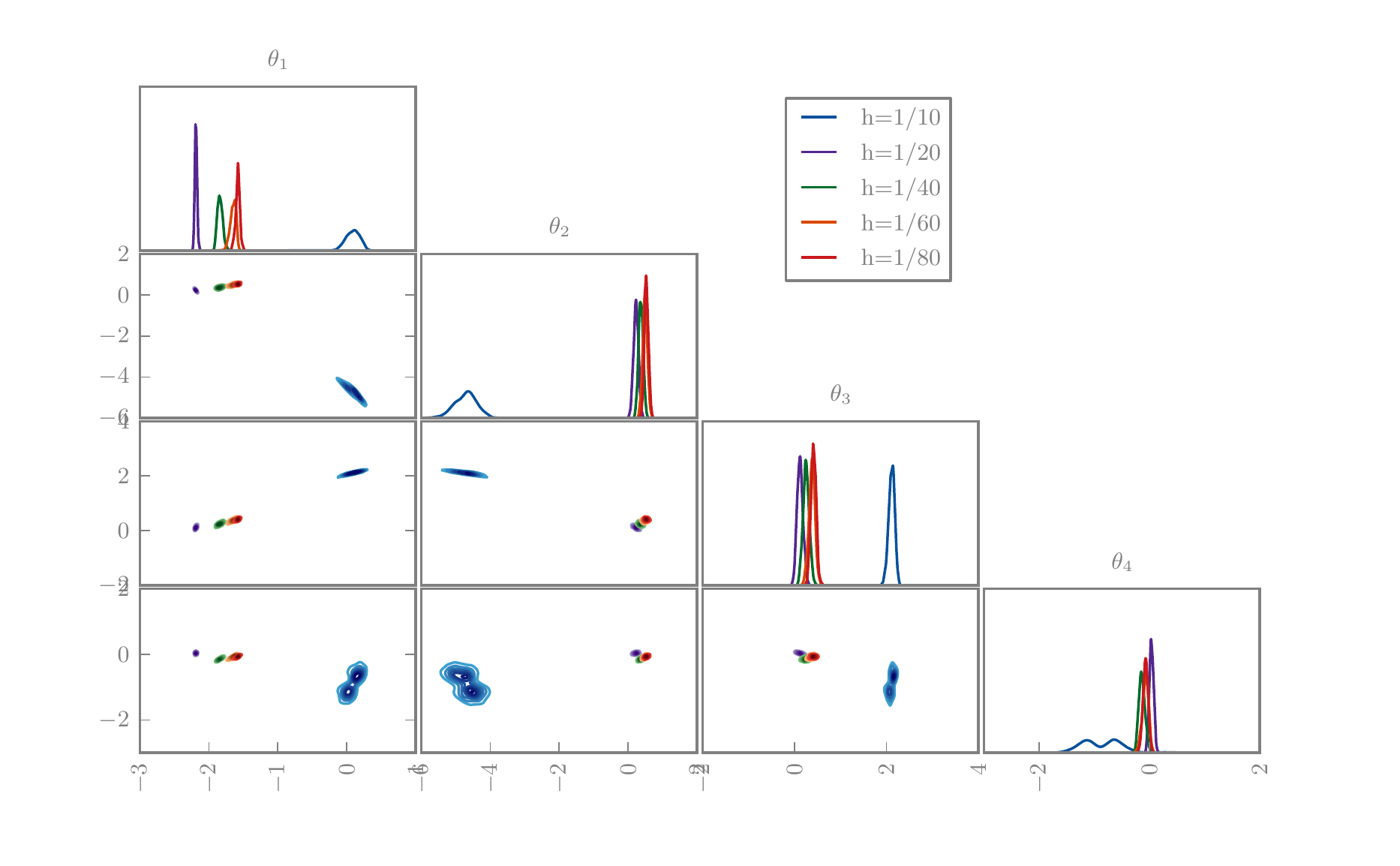}

\caption{The marginal posterior distributions for the first four coefficients in 1D elliptic inverse problem using a classic deterministic solver with various grid sizes.}
\label{fig:elliptic1d_det}
\end{figure}

\begin{figure}[p]
\centering
\includegraphics[scale=.6, trim=0in .25in 0in .25in]{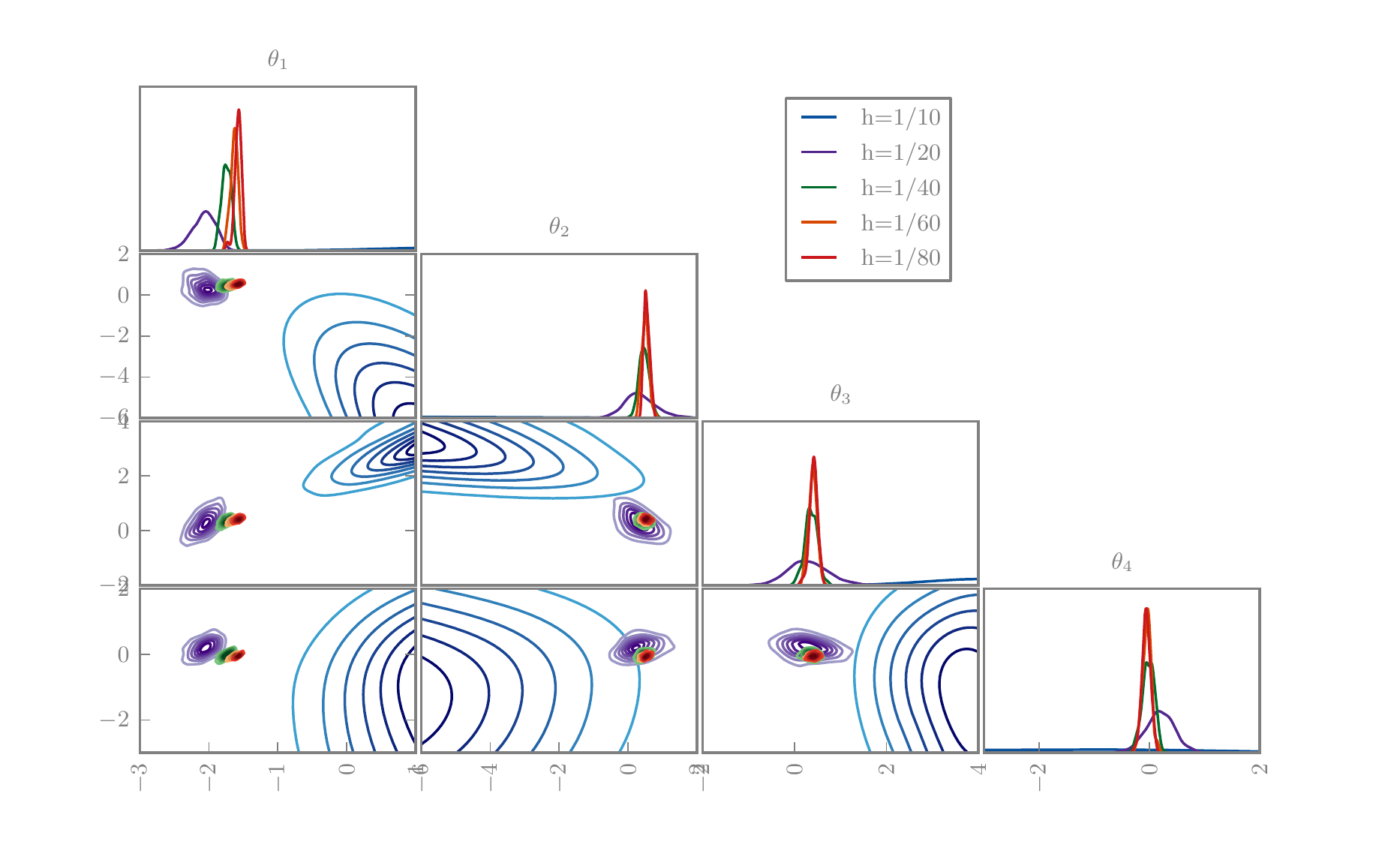}

\caption{The marginal posterior distributions for the first four coefficients in 1D elliptic inverse problem using a randomised solver with various grid sizes.}
\label{fig:elliptic1d_rand}
\end{figure}

\section{Conclusions}
We have presented theory and methods for probabilistic approaches to the 
numerical solution of both ordinary and partial differential equations. 
These methods give rise to a probability measure whose 
qualitative and quantitative properties account for and quantify the uncertainty induced by finite dimensional approximation. We provide a theoretical analysis of the properties of these probabilistic integrators and demonstrate that they induce more coherent inferences in illustrative empirical examples. This work opens the door for statistical analysis that explicitly incorporates numerical uncertainty in many important classes of contemporary statistical problems across the sciences, including engineering, climatology, and biology.

Drawing parallels to model error \cite{kennedy2001bayesian}, 
we can consider our intrusive modifications to be a highly-specialised discrepancy model, designed using our intimate knowledge of the structure and properties of numerical methods. Further study is required to compare our approach to existing methods, and we hope to develop other settings where modifying the internal structure of numerical methods can produce computationally and analytically tractable measures of uncertainty. 
Developing robust, practical solvers, and efficiently performing computations with the additional uncertainty we propose remains a challenge. Finally, future work may be able to exploit our coherent uncertainty propagation to optimise the tradeoff between solver cost and statistical variance.

\bibliographystyle{imsart-nameyear}
\bibliography{citations.bib}

\end{document}